\documentclass[11pt,oneside,english]{article}
\pdfoutput=1
\usepackage{amssymb,amsmath,amsfonts,mathtools,amsthm,euscript,mathrsfs,verbatim,enumerate,multirow,bbding,slashed,multicol,color,array,esint,babel,tikz,
tikz-cd,tikz-3dplot,
pgfplots,ytableau,rotating,hyperref,geometry,mathdots,savesym,cite,wasysym,amscd,graphicx,pifont,float,setspace,wrapfig,picture,subfigure,tabularx,datetime,cleveref,makecell,adjustbox,youngtab,parskip,cases,dynkin-diagrams,chronology,longtable}
\usepackage[utf8]{inputenc}
\usepackage[T1]{fontenc}
\usepackage{enumitem}
\usepackage[toc,title]{appendix}
\numberwithin{equation}{section}
\usetikzlibrary{arrows,positioning,decorations.pathmorphing, decorations.markings, matrix, patterns, shapes}
\hypersetup{colorlinks=true}
\hypersetup{linkcolor=black}
\hypersetup{citecolor=black}
\hypersetup{urlcolor=black}
\DeclareMathOperator{\CS}{CS}
\DeclareMathOperator{\Tr}{Tr}
\DeclareMathOperator{\ad}{ad}
\DeclareMathOperator{\adj}{adj}
\DeclareMathOperator{\Ad}{Ad}
\newcommand{\trb}{\mathbf{tr}}
\theoremstyle{plain}
\newtheorem*{thm*}{Theorem}
\theoremstyle{plain}% default
\newtheorem{thm}{Theorem}[section]

\newtheorem{prop}[thm]{Proposition}
\newtheorem{proposal}[thm]{Proposal}
\theoremstyle{definition}
\newtheorem{defn}[thm]{Definition}

\newtheorem{notation}[thm]{Notation}
\newtheorem{rem}[thm]{Remark}

\tikzset{
  big arrow/.style={
    decoration={markings,mark=at position 1 with {\arrow[scale=1.5,#1]{>}}},
    postaction={decorate},
    shorten >=0.4pt},
  big arrow/.default=black}

\allowdisplaybreaks
\pgfplotsset{compat=1.18}
\begin{document}

\begin{titlepage}
\begin{center}
\vspace{1.8cm}
{\Huge\bfseries 
The topological life of Dynkin indices: universal scaling and matter selection\\
 } 
\vspace{2.5cm}
{\Large
Mboyo Esole$^{\diamond}$ and Monica Jinwoo Kang$^{\clubsuit}$\\}
\vspace{.6cm}
{\large $^{\diamond}$ Department of Mathematics, Northeastern University}\par
{\large   Boston, MA 02115, U.S.A.}\par
\vspace{.2cm}
{\large $^\clubsuit$ 
Mitchell Institute for Fundamental Physics and Astronomy, Texas A\&M University}\par
{College Station, TX 77843, U.S.A}\par
\vspace{.6cm}
{\tt j.esole@northeastern.edu, monicak6@tamu.edu}\par
\vspace{1.6cm}
\end{center}

\begin{abstract}
\noindent For simple, simply-connected compact Lie groups, Dynkin embedding indices obey a universal scaling law with a direct topological meaning. Given an inclusion $f:G\hookrightarrow H$, the Dynkin embedding index $j_f$ is characterized equivalently by the induced maps on $\pi_3$ and on the canonical generators of $H^3$, $H^4(B{-})$, and $H^4(\Sigma{-})$. Consequently, $j_f$ controls instanton-number scaling, the quantization levels of Chern--Simons and Wess--Zumino--Witten terms, and the matching of gauge couplings and one-loop RG scales. We connect this picture to representation theory via the $\beta$-construction in topological $K$-theory, relating Dynkin indices to Chern characters through Harris' degree--$3$ formula and Naylor's suspended degree--$4$ refinement. Finally, we apply these results to F-theory to explain the prevalence of index-one matter: we propose a ``genericity heuristic'' where geometry favors regular embeddings (typically $j_f=1$) associated with minimal singularity enhancements, while higher-index embeddings require non-generic tuning.
\end{abstract}

\vfill 
%{ \today\   at \ \currenttime\par}

\end{titlepage}

\tableofcontents

\newpage 

\section{Introduction}\label{sec:Intro}

Matter fields in gauge theories are organized by finite-dimensional irreducible representations of a Lie algebra $\mathfrak g$. They are called \emph{matter representations}. In geometric compactifications of M-theory and F-theory on elliptic fibrations, one obtains a rich spectrum of such representations which, until recently, lacked a clear intrinsic criterion distinguishing them from the larger universe of possible representations \cite{EK11}.

A powerful proposal due to Katz and Vafa \cite{Katz:1996xe} is to deduce the matter representations by embedding $\mathfrak g$ into a larger ``parent'' algebra $\mathfrak m$ and decomposing the adjoint representation of $\mathfrak m$ along $\mathfrak g$. As explained in \cite{EK11}, this \emph{Katz--Vafa method} is ambiguous if one specifies $\mathfrak g$ and $\mathfrak m$ only up to isomorphism: in general there are many distinct embeddings $\mathfrak g\hookrightarrow \mathfrak m$ producing different branching rules for the adjoint representation of $\mathfrak m$.

In \cite{EK11} we explained that, to control this ambiguity, Dynkin's refined theory of linear equivalence of subalgebras of simple Lie algebras becomes essential, together with the notions of embeddings and linear equivalence of subalgebras, characteristic representations, and in particular the \emph{Dynkin embedding index}. We review these notions in Section~\ref{sec:DI}. We argued that up to multiplicities and trivial summands, the matter representations appearing in generic F-theory compactifications are precisely the \emph{characteristic representations} of certain embeddings singled out by a simple index-one condition.

\medskip
\begin{proposal}[Dynkin index-one matter selection rule {\cite{EK11}}]\label{SR}
In generic F-theory configurations, up to multiplicities and trivial representations, the matter representations with respect to a Lie algebra $\mathfrak g$ arise as characteristic representations of embeddings
$
\mathfrak g\hookrightarrow \mathfrak m
$
whose Dynkin embedding index equals $1$ along each simple component of $\mathfrak g$.
\end{proposal}

\noindent\textbf{Problem Statement.}
A fundamental question left open in \cite{EK11} is the \emph{physical meaning} of this condition: why should index-one embeddings be privileged from the viewpoint of gauge theory?

\noindent\textbf{A topological perspective.}
To address this question, we start by looking at the problem through the lens of the algebraic-topology description of charges and topological terms in the physics of principal bundles. In doing so, we consider the simple simply-connected compact group associated to a given simple Lie algebra.

We will need some tools to be able to make meaningful statements, and we review the needed background in the appendices. For any space $X$, we denote by $\Sigma X$ its suspension and by $\Omega X$ its loop space. Given two spaces $X$ and $Y$, $[X,Y]$ denotes the set of homotopy classes of continuous maps from $X$ to $Y$. Given a topological group $G$, we denote by $BG$ its classifying space. There is a universal principal $G$--bundle
\begin{equation*}
G\longrightarrow EG\longrightarrow BG
\end{equation*}
with total space $EG$ contractible, and one has a homotopy equivalence $\Omega BG\simeq G$.

Principal $G$--bundles over a paracompact space $X$ are classified up to isomorphism by maps $X\to BG$ up to 
homotopy, that is, by the set $[X,BG]$. The group $\pi_3(G)\cong \mathbb Z$ governs instanton charge (equivalently $\pi_4(BG)$), while $H^4(BG,\mathbb Z)$ and $H^3(G,\mathbb Z)$ govern Chern--Simons and Wess--Zumino--Witten terms \cite{DW, Henriques}. In symmetry breaking of a Lie group $H$ to a subgroup $G$, the quotient $H/G$ is a homogeneous space, the projection $H\to H/G$ is a principal $G$--bundle, and the homotopy group $\pi_3(H/G)$ provides key information on instantons in the symmetry-breaking phase of the theory \cite{Csaki}. Armed with these tools, we can have a more refined look at matter representations and the role of the Dynkin embedding index. Connecting through the $K$--theory of compact Lie groups will also give an elegant overall picture of the common theme in topology unifying these different objects.

\subsection{Dynkin index-one matter selection rule as a topological conservation law}
In this paper we show that the Dynkin embedding index is the \emph{unique} integer controlling the simultaneous rescaling of the basic topological charges and quantization levels in gauge theory under embeddings of structure groups.

For a simple, simply-connected compact Lie group $G$, classical results of Cartan, Bott, Koszul, Hopf, and Borel imply
\begin{equation*}
\pi_3(G)\cong \mathbb Z,\qquad H^3(G,\mathbb Z)\cong \mathbb Z,\qquad H^4(BG,\mathbb Z)\cong \mathbb Z,\qquad H^4(\Sigma G,\mathbb Z)\cong \mathbb Z,
\end{equation*}
where $BG$ is the classifying space of $G$ while $\Sigma G$ is the suspension of $G$.
Cohomology suspension and transgression (via the Serre spectral sequence of $G\to EG\to BG$) identify compatible generators across these groups.
As each of these groups is infinite cyclic, each has exactly two generators related by a change of sign.

If $f:G\hookrightarrow H$ is an embedding between simple, simply-connected compact Lie groups, it induces homomorphisms
\begin{align*}
(Bf)^* &: H^4(BH,\mathbb Z)\longrightarrow H^4(BG,\mathbb Z),\\
f^* &: H^3(H,\mathbb Z)\longrightarrow H^3(G,\mathbb Z),\\
f_* &: \pi_3(G)\longrightarrow \pi_3(H),\\
(\Sigma f)^* &: H^4(\Sigma H,\mathbb Z)\longrightarrow H^4(\Sigma G,\mathbb Z).
\end{align*}
Since each target is infinite cyclic, each map is multiplication by an integer.
After fixing the standard normalization and making compatible choices of generators, \emph{the same positive integer appears in every case}, and it is exactly the Dynkin embedding index $j_f$.
Dynkin first proved that the Dynkin index appears as the multiplicative factor for the map $f^*$ in the case of classical groups \cite{Dynkin59}.
His result was generalized to all simply-connected simple compact Lie groups by Arkady~L.~Onishchik \cite{Oniscik60}.

The appearance of the same index in all these cases is particularly important because these maps have deep physical meaning: under an embedding $f: G\to H$, the spectra of
\begin{itemize}
\item instanton charges (governed by $\pi_3$),
\item Chern--Simons levels (classified by $H^4(BG,\mathbb Z)$), and
\item WZW levels (classified by $H^3(G,\mathbb Z)$)
\end{itemize}
are simultaneously rescaled by $j_f$.

Thus, the Dynkin embedding index $j_f$ equals $1$ if and only if these charges/levels are preserved.
If $j_f>1$, the minimal charge in $H$ is no longer minimal in $G$ and could lead to a fractionalization of the charge unit. That is exactly the phenomenon studied in \cite{Csaki} in the context of symmetry breaking of a group $H$ to a subgroup $G$.

This leads us to the following proposal.

\begin{proposal}[Topological conservation principle]\label{prop:TopologicalConservation}
Let $f:G\hookrightarrow H$ be an embedding of simple, simply-connected compact Lie groups with Dynkin embedding index $j_f$.
After compatible choices of generators, the induced maps on $\pi_3$, $H^3$, $H^4(B{-})$, and $H^4(\Sigma{-})$
are all multiplication by $j_f$. In particular $j_f=1$ if and only if these charges/levels are preserved
without fractionalization.
\end{proposal}

When one restricts structure group along an embedding $f:G\hookrightarrow H$, for example in the analysis of instanton effects after partial breaking \cite{Csaki,Intriligator:1995id}, one needs to track how these integer normalizations change.

Since $G$ is compact, $f(G)\subset H$ is closed, so the homogeneous space $H/G$ is a smooth manifold and $H\to H/G$ is a principal $G$--bundle.
The same scaling implies the familiar discrete remnant
\begin{equation*}
\pi_3(H/G)\cong \mathbb Z/j_f\mathbb Z,
\end{equation*}
as discussed in detail in \cite{Csaki} and proved here as Proposition \ref{prop:pi3_H_over_G}. 
In particular, $\pi_3(H/G)$ is trivial {\em if and only if} $j_f=1$.

The topological conservation principle in Proposal~\ref{prop:TopologicalConservation}
is not an abstract coincidence: it is mirrored across several foundational analyses of
instanton physics \cite{Bernard,AHS,Intriligator:1995id,Csaki}. 
F Atiyah--Hitchin--Singer identify $\pi_{3}(G)\cong\mathbb Z$ (equivalently
$H^{4}(BG,\mathbb Z)\cong\mathbb Z$ via characteristic classes) as the basic
topological invariant organizing self-dual Yang--Mills sectors (see also \cite{BelavinEtAl1975BPST,Bernard}, and then they use
index-one subgroup chains to transport solutions without changing their charge
\cite{AHS}.
Second, Intriligator--Seiberg emphasize that in spontaneous symmetry breaking
$H\to G$, a sharp physical distinction between $j_f=1$ and $j_f>1$ is
precisely whether there exist genuine UV instanton sectors that cannot be realized
inside the unbroken subgroup and hence must be accounted for separately when the
massive fields are integrated out \cite[Footnote~2]{Intriligator:1995id}.
Third, Cs\'aki--Murayama identify the coset homotopy group $\pi_3(H/G)$ as the
detector of these additional sectors and relate it to the same Dynkin index,
yielding $\pi_{3}(H/G)\cong\mathbb Z/j_f\mathbb Z$ \cite{Csaki}.
We synthesize these perspectives in Section~\ref{sec:instanton_story}, where they
form a coherent physical interpretation of the index-one condition as a
conservation law for quantized topological charges under embeddings.

\subsection{Topology and K-theory behind the scaling}

A second theme of the paper is that the Dynkin index of a representation is detected by the Chern character in complex topological $K$-theory.
$K$-theory allows one to see the universal rescaling in action from the point of view of the Dynkin index of a representation.

Complex topological $K$-theory was developed in the early 1960s by Michael Atiyah and Friedrich Hirzebruch, inspired by the study of complex vector bundles and Bott periodicity. It is $\mathbb Z/2$--graded:
\begin{equation*}
K^*(X)=K^0(X)\oplus K^{-1}(X).
\end{equation*}
We associate to any compact Lie group $G$ its representation ring $R(G)$.
The $\beta$--construction gives a morphism from the representation ring $R(G)$ to the odd $K$--theory group $K^{-1}(G)$:
\begin{equation*}
\beta:\ R(G)\longrightarrow K^{-1}(G)\cong [G,U],
\end{equation*}
where $U$ denotes the stable unitary group, and $[G,U]$ denotes homotopy classes of continuous maps.
Via Bott periodicity, we also obtain a suspended class $\beta_\Sigma(\rho)\in \widetilde K^0(\Sigma G)$.

Harris' theorem identifies the Dynkin index $\ell_\rho$ of a representation $\rho$ as the coefficient of the generator in degree $3$ via the \emph{odd} Chern character \cite{Harris}:
\begin{equation*}
\operatorname{ch}_3(\beta(\rho))=\ell_\rho\,x_3\in H^3(G,\mathbb Z),
\end{equation*}
where the generator $x_3$ is chosen so that $\ell_\rho$ is positive.
Naylor's suspension-based description gives the complementary degree-$4$ statement (see Theorem~\ref{thm:Naylor_corrected}):
\begin{equation*}
\operatorname{ch}_2(\beta_\Sigma(\rho))=\ell_\rho\,u_4\in H^4(\Sigma G,\mathbb Z),
\qquad u_4=\Sigma^\sharp(x_3).
\end{equation*}
This functorial package is summarized in Figure~\ref{fig:hexagon}, which we use throughout to track how a single
integer controls $\pi_3$, $H^3$, $H^4(B-)$, $H^4(\Sigma-)$, and the $K$--theory classes attached to representations.

If $f:G\hookrightarrow H$ is an embedding, then naturality of the Chern character means precisely that for any
map $g:X\to Y$ and any $a\in K^*(Y)$ one has
\begin{equation*}
\operatorname{ch}(g^!a)=g^*\operatorname{ch}(a).
\end{equation*}
Applying this to $f$ (and to $\Sigma f$) shows that restriction of representations multiplies Dynkin indices by the same
integer that rescales $x_3$, $y_4$, and $u_4$. In this way, this integer can be read off equally from $\pi_3$, from
cohomology, or from the Chern-character images in $K$-theory. Concretely, after fixing compatible generators
\begin{equation*}
x_3(H)\mapsto j_f\,x_3(G),\qquad y_4(H)\mapsto j_f\,y_4(G),\qquad u_4(H)\mapsto j_f\,u_4(G),
\end{equation*}
one has for every representation $\rho$ of $H$
\begin{equation*}
\ell_{\rho\circ f}=j_f\,\ell_\rho,
\end{equation*}
and the same factor governs instanton-number scaling via $\pi_3$ (equivalently $\pi_4(B{-})$), the
transgression $H^3(G)\to H^4(BG)$, and the suspension class in $H^4(\Sigma G)$.

It follows that the Dynkin index of the representation $\rho$ of $H$ agrees with that of its restriction $\rho\circ f$ to $G$
if and only if the Dynkin embedding index is $j_f=1$.
This ``universal scaling'' viewpoint is the organizing principle of the paper and the mechanism behind the
selection rule that motivated \cite{EK11}.
\subsection{Organization of the paper}

The paper is organized as follows.

Section~\ref{sec:DI} fixes conventions for invariant symmetric bilinear forms on simple Lie algebras, including the reduced Killing form normalization (and the compact sign convention), and reviews Dynkin indices of representations and Dynkin embedding indices of inclusions.  We also recall the intrinsic group invariant $d_G$ (the gcd of Dynkin indices), and record the basic numerical data needed later (Table~\ref{Table:DataSimpleGroup} and Figure/Table~\ref{DynkinT5}).

Section~\ref{sec:instanton_story} explains why the dichotomy $j_f=1$ versus $j_f>1$ is physically meaningful in gauge theory.  We synthesize the viewpoints of Atiyah--Hitchin--Singer, Intriligator--Seiberg, and Cs\'aki--Murayama to show that $j_f=1$ is precisely the condition under which the minimal instanton sector (and hence the basic charge normalization) is preserved under embeddings and symmetry breaking, while $j_f>1$ detects genuinely new coset sectors and “fractionalization” phenomena.  We also discuss the dynamical matching of couplings and scales, and briefly indicate the parallel normalization issues that appear in string compactifications.

Section~\ref{sec:Ftheory_genericity} provides a geometric heuristic explaining why index-one embeddings occur frequently in standard F-theory engineering: codimension-two enhancements in minimal Weierstrass/Tate tunings often realize the smallest available embedding index, and in many common families that smallest index is $1$.  We also emphasize the distinction between regular versus special subalgebras and outline where $j_f>1$ can plausibly arise in more constrained (non-generic) corners.

Section~\ref{sec:K} develops the $K$-theoretic bridge between representation theory and topology.  We review the $\beta$-construction, Bott suspension, the odd Chern character, and the representation-theoretic structure of $K(G)$.  We then state Harris’ degree-$3$ formula and Naylor’s suspended degree-$4$ refinement, and package these compatibilities into the functorial “hexagon” of Figure~\ref{fig:hexagon}.

Section~\ref{sec:universalscaling} proves the universal scaling theorem: for an embedding $f:G\hookrightarrow H$ of simple, simply-connected compact Lie groups, a single integer rescales in a compatible way the generators in $H^{3}(-)$, $H^{4}(B-)$, $H^{4}(\Sigma-)$, and $\pi_{3}(-)$, and this integer is exactly the Dynkin embedding index $j_f$.  We also derive $\pi_3(H/G)\cong \mathbb Z/j_f\mathbb Z$ and explain how naturality in $K$-theory forces the corresponding scaling law for Dynkin indices under restriction of representations, thereby identifying the index-one selection rule with a conservation law for the basic quantized topological charges and levels.

The appendices collect background and normalizations used throughout.  Appendix~\ref{sec:A} summarizes the low-degree topology of $G$ and $BG$, the suspension/loop formalism, and the transgression/suspension compatibilities that organize the generators $x_3$, $y_4$, and $u_4$.  Appendix~\ref{sec:AppB} records the compact-sign conventions and the highest-root $SU(2)$ normalization used to pin down generators and the Cartan $3$-form.

\medskip
\noindent\textbf{Reading guide.}
Readers primarily interested in the physical meaning of the index-one condition may read
Section~\ref{sec:DI} $\rightarrow$ Section~\ref{sec:instanton_story} $\rightarrow$ Section~\ref{sec:universalscaling}.
Readers focused on the $K$-theoretic interpretation may read
Section~\ref{sec:K} (especially Theorems~\ref{thm:Harris} and \ref{thm:Naylor_corrected} and Figure~\ref{fig:hexagon})
before moving to Section~\ref{sec:universalscaling}.
Readers approaching from F-theory may read Section~\ref{sec:Ftheory_genericity} after the definitions in
Section~\ref{sec:DI}, and then consult Sections~\ref{sec:instanton_story} and~\ref{sec:universalscaling} for the
topological mechanism underlying the index-one heuristic.

\section{Invariant symmetric bilinear forms}\label{sec:DI}

Let $\mathfrak{g}$ be a semisimple Lie algebra over $\mathbb{C}$ with Lie bracket $[\cdot,\cdot]$.
A symmetric bilinear form $B(\cdot,\cdot):\mathfrak{g}\times\mathfrak{g}\to\mathbb{C}$ is said to be
{\em $\operatorname{Ad}$-invariant} (or {\em invariant}) if
\begin{equation}\label{eq:ad-invariance-lie}
B([X,Y],Z) + B(Y,[X,Z]) = 0 \qquad \text{for all }X,Y,Z\in\mathfrak{g}.
\end{equation}
Equivalently, if $G$ is a connected Lie group with Lie algebra $\mathfrak{g}$, invariance means
\begin{equation*}
B(\operatorname{Ad}_g X,\operatorname{Ad}_g Y)=B(X,Y)\qquad \text{for all }g\in G,\ X,Y\in\mathfrak g.
\end{equation*}

\begin{rem}[Semisimple versus simple]\label{rem:semisimple-forms}
If $\mathfrak g=\bigoplus_{a=1}^s \mathfrak g_a$ is semisimple with simple ideals $\mathfrak g_a$,
then the space of invariant symmetric bilinear forms on $\mathfrak g$ is $s$-dimensional and is spanned
by the Killing forms of the simple summands.  In particular, when $\mathfrak g$ is \emph{simple} the space of
invariant forms is one-dimensional, so any two invariant symmetric bilinear forms are proportional.
\end{rem}

A basic example is $\mathfrak{g}=\mathfrak{sl}(n,\mathbb{C})$, where the trace of the product of two matrices
\begin{equation*}
(X,Y)\longmapsto \operatorname{Tr}(XY)
\end{equation*}
is an invariant symmetric bilinear form.
More generally, if $f:\mathfrak{g}\to\mathfrak{sl}(V)$ is a finite-dimensional representation of a Lie algebra
$\mathfrak g$, then
\begin{equation}\label{eq:trace-form}
(X,Y)_f := \operatorname{Tr}\big(f(X)\,f(Y)\big)
\end{equation}
defines an invariant symmetric bilinear form on $\mathfrak g$.
These trace forms are the basic invariant quadratic inputs that appear in Chern--Weil theory
(via invariant quadratic polynomials on $\mathfrak g$).

\subsection{The Killing form and its normalization}

Let $\mathfrak{g}$ be a complex Lie algebra. The adjoint representation is
\begin{equation*}
\operatorname{ad}:\mathfrak{g}\to \mathfrak{gl}(\mathfrak{g}),\qquad X\mapsto \operatorname{ad}_X,
\quad\text{where }\operatorname{ad}_X(Y)=[X,Y].
\end{equation*}
The Jacobi identity is equivalent to the statement that $\operatorname{ad}_X$ is a derivation:
\begin{equation*}
\operatorname{ad}_X([Y,Z])=[\operatorname{ad}_X(Y),Z]+[Y,\operatorname{ad}_X(Z)].
\end{equation*}
The \emph{Killing form} is the trace form of the adjoint representation:
\begin{equation}\label{eq:killing}
\kappa(X,Y)=\operatorname{Tr}(\operatorname{ad}_X \circ \operatorname{ad}_Y).
\end{equation}
Cartan's criterion states that $\kappa$ is nondegenerate if and only if $\mathfrak g$ is semisimple.
For a compact real form of a semisimple Lie algebra, $\kappa$ is negative definite.

From now on we assume that $\mathfrak g$ is \emph{simple}. By Remark~\ref{rem:semisimple-forms},
any invariant symmetric bilinear form on $\mathfrak g$ is a scalar multiple of $\kappa$.
Let $\mathfrak h\subset\mathfrak g$ be a Cartan subalgebra and let $\Phi\subset\mathfrak h^*$ be the corresponding
root system. Denote by $\theta$ the highest root.

\begin{defn}[Reduced Killing form]\label{defn:reduced-killing}
The \emph{reduced Killing form} is the invariant symmetric bilinear form $(\cdot,\cdot)$ defined by
\begin{equation}\label{eq:reducedKilling}
(X,Y):=\frac{1}{2h^\vee}\,\kappa(X,Y),
\end{equation}
where $h^\vee$ is the dual Coxeter number of $\mathfrak g$.
Via the induced bilinear form on $\mathfrak h^*$, every long root has squared length equal to $2$.
In particular,
\begin{equation}\label{eq:theta-length}
(\theta,\theta)=2.
\end{equation}
\end{defn}

\begin{rem}\label{rem:compact_sign}
For the compact real form, the Killing form is negative definite; in that setting we use the sign-reversed
normalization so that the corresponding inner product is positive definite. See Appendix~\ref{sec:AppB}.
\end{rem}

For any finite-dimensional representation $\rho:\mathfrak g\to \mathfrak{gl}(V)$, the trace form
$(X,Y)\mapsto \Tr\big(\rho(X)\rho(Y)\big)$ is invariant and symmetric, hence there exists a scalar $\ell_\rho$ such that
\begin{equation}\label{eq:trace_index}
\Tr\big(\rho(X)\rho(Y)\big)=\ell_\rho\,(X,Y).
\end{equation}

\subsection{Dynkin index of a linear representation}

Let $\mathfrak g$ be simple and equip it with the reduced Killing form $(\cdot,\cdot)$ from \eqref{eq:reducedKilling}.
Let $f:\mathfrak g\to \mathfrak{gl}(V)$ be a finite-dimensional representation.
Then \eqref{eq:trace-form} defines an invariant symmetric bilinear form $(\cdot,\cdot)_f$ on $\mathfrak g$.
Since $\mathfrak g$ is simple, there exists a unique scalar $\ell_f\ge 0$ such that
\begin{equation}\label{eq:def-ellf}
\operatorname{Tr}\big(f(X)f(Y)\big)=\ell_f\,(X,Y)\qquad (X,Y\in \mathfrak g).
\end{equation}

\begin{defn}\label{defn:Dynkin-index-rep}
The scalar $\ell_f$ is called the \emph{Dynkin index} of the representation $f$.
\end{defn}

Clearly $\ell_f=0$ if and only if $f$ is the trivial representation.
If $f$ is irreducible with highest weight $\lambda_f$, Dynkin showed that \cite[Theorem 2.5]{Dynkin.SubA}
\begin{equation}\label{eq:DynkinIndexHighestWeight}
\ell_f=\frac{\dim f}{\dim \mathfrak g}\,(\lambda_f,\lambda_f+2\rho),
\end{equation}
 where $\rho$ is the Weyl vector (half the sum of positive roots).
With our normalization of $(\cdot,\cdot)$, one has $\ell_f\in \mathbb Z_{\ge 0}$ for every irreducible
representation (see Dynkin; modern accounts include Gross and related references).

The Killing form $\kappa$ corresponds to the adjoint representation, so \eqref{eq:reducedKilling} gives
\begin{equation*}
\kappa(X,Y)=2h^\vee\,(X,Y),
\end{equation*}
and hence the Dynkin index of the adjoint representation is
\begin{equation}\label{eq:ell-adj}
\ell_{\mathrm{adj}}=2h^\vee.
\end{equation}
Comparing \eqref{eq:DynkinIndexHighestWeight} with the adjoint case (highest weight $\theta$ and $\dim f=\dim\mathfrak g$)
also yields the standard identity
\begin{equation}\label{eq:dual-coxeter-theta-rho}
h^\vee=1+(\theta,\rho).
\end{equation}

\subsection{Dynkin index of an embedding}

Let $f:\tilde{\mathfrak m}\hookrightarrow \mathfrak m$ be an injective homomorphism between simple Lie algebras.
Choose Cartan subalgebras $\tilde{\mathfrak h}\subset \tilde{\mathfrak m}$ and $\mathfrak h\subset \mathfrak m$
such that $f(\tilde{\mathfrak h})\subset \mathfrak h$.
Let $(\cdot,\cdot)_{\tilde{\mathfrak m}}$ and $(\cdot,\cdot)_{\mathfrak m}$ denote the reduced Killing forms on
$\tilde{\mathfrak m}$ and $\mathfrak m$, each normalized so that the highest root in the corresponding root system
has squared length $2$.

\begin{defn}\label{defn:embedding-index}
The \emph{Dynkin embedding index} of $f$ is the number $j_f$ such that
\begin{equation}\label{eq:def-jf}
(f(X),f(Y))_{\mathfrak m}=j_f\,(X,Y)_{\tilde{\mathfrak m}}, \quad \text{for all}\quad X,Y\in \tilde{\mathfrak m}.
\end{equation}
\end{defn}

Let $\tilde\theta$ be the highest root of $\tilde{\mathfrak m}$, by definition of the reduced Killing form we have 
$(\tilde\theta,\tilde\theta)_{\tilde{\mathfrak m}}=2$.
Then \eqref{eq:def-jf} implies
$(f(\tilde\theta),f(\tilde\theta))_{\mathfrak m}=j_f\,(\tilde\theta,\tilde\theta)_{\tilde{\mathfrak m}}=2j_f$. We conclude that (see \cite[Equation 2.3]{Dynkin.SubA})
\begin{equation}\label{eq:jf-half}
j_f=\frac12\,(f(\tilde\theta),f(\tilde\theta))_{\mathfrak m}.
\end{equation}
Dynkin proved that in fact $j_f$ is always a positive integer. 

\begin{thm}[See Theorem 2.2 of  \cite{Dynkin.SubA}]\label{thm:integrality_jf}
The Dynkin embedding index $j_f$ is a positive integer.
\end{thm}

Dynkin's original proof \cite{Dynkin.SubA} proceeds by a detailed Lie--theoretic analysis of $\mathfrak{sl}_2$-subalgebras
(equivalently, of $A_1$-embeddings) inside each simple Lie algebra type.
There are also conceptual proofs: topologically, once one identifies $j_f$ with the integer by which $f$
scales the canonical generator of $H^3(G,\mathbb Z)\cong \mathbb Z$ (equivalently, the degree on
$\pi_3(G)\cong \mathbb Z$), integrality becomes automatic.  We develop this topological interpretation in
Section~\ref{sec:universalscaling}. For purely algebraic proofs, see \cite[Chap 3. \S 2.9]{OnishchikVinberg:LGIII} and \cite{Braden:1991}.

\subsection{The Dynkin index of a simple Lie group}\label{subsec:DI_dG}

In addition to the Dynkin index $\ell_\rho$ of a representation $\rho$ and the Dynkin embedding index $j_f$,
we also consider an intrinsic integer attached to a compact Lie group \cite{Dynkin59,LaszloSorger,Harris,Naylor}.

\begin{defn}[Dynkin index of a simple Lie group {\cite{LaszloSorger,Harris}}]\label{def:dG} 
Let $G$ be a compact Lie group. Fix the normalization of the invariant bilinear form on $\mathfrak g$ so that long roots have squared length $2$. The \emph{Dynkin index of $G$} is the integer
\begin{equation}
d_G \;:=\; \gcd\{\ell_\rho \;:\; \rho \text{ a finite-dimensional representation of }G\}.
\end{equation}
Equivalently, since the Dynkin index is additive under direct sums ($\ell_{V\oplus W} = \ell_V + \ell_W$), the set of all indices forms an ideal in $\mathbb{Z}$; $d_G$ is the unique positive generator of this ideal.
\end{defn}

\begin{rem}[Historical note: Dynkin and Onishchik]
The integer $d_G$ was implicitly identified by Dynkin for the classical groups in his study of the topological characteristics of homomorphisms \cite{Dynkin59}. By analyzing the map $\pi_3(G) \to \pi_3(U(N)) \cong \mathbb{Z}$ induced by the standard representations, Dynkin found that the minimal topological degree is $1$ for $SU(n)$ and $Sp(n)$, but equals $2$ for the spin groups $Spin(n)$ ($n\ge 7$). This reflects the fact that the vector representation of $Spin(n)$ carries two units of the intrinsic topological charge normalized by the highest root. The topological interpretation was later generalized to all simple Lie groups (including the exceptional cases) by Dynkin's student Arkady L. Onishchik \cite{Oniscik60}, who established the integrality of the index using cohomological methods.
\end{rem}

\paragraph{Computing $d_G$ from finitely many representations}
In practice one may compute $d_G$ as the gcd of the Dynkin indices of the fundamental representations.
This reduction is standard (see \cite{LaszloSorger}) in view of the tensor-product identity satisfied by the Dynkin index \cite[\S 2 Equation 2.25]{Dynkin.SubA}:
\begin{equation}
\ell_{V\otimes W}=\ell_V\dim(W)+\ell_W\dim(V),
\end{equation}
together with the fact that every finite irreducible representation occurs as a summand of a tensor product of
fundamental representations.

\medskip

\begin{rem}[Other incarnations of $d_G$]\label{rem:dG_other_occurrences}
The integer $d_G$ appears in several equivalent guises.

\noindent\emph{(i) $K$--theory / odd Chern character (Harris).}
Harris \cite{Harris} establishes that the image of the Chern character $\operatorname{ch}_3: K^{-1}(G) \to H^3(G,\mathbb{Z})$ is the subgroup generated by $d_G\,x_3(G)$, where $x_3(G)$ is the canonical generator of $H^3(G,\mathbb{Z})$; in this framework, $d_G$ appears as the index of the image lattice in the integral cohomology.

\noindent\emph{(ii) Suspended Chern character (Naylor).}
Naylor \cite{Naylor} identifies the image of the relevant low-degree Chern character with
$d_G\,\mathbb Z$ inside the target infinite cyclic group (see Theorem~\ref{thm:Naylor_corrected}); thus $d_G$ measures the precise divisibility defect of the representation-theoretic map.

\noindent\emph{(iii) Braden's pairing matrix.}
The same integer $d_G$ appears algebraically as the gcd of the entries of Braden's matrix
$H_{ij}$, defined by an integral pairing on dominant weights evaluated on the fundamental weights (cf.\ \cite[Table~1]{Braden:1991}).

\noindent\emph{(iv) Determinant line bundles on moduli.}
Laszlo--Sorger introduce $d_G$ in their study of determinant line bundles on moduli of $G$--bundles;
their ``Dynkin index of $G$'' agrees with Definition~\ref{def:dG} \cite{LaszloSorger}.
\end{rem}

\paragraph{Relation to Chern--Weil traces.}
Let $P\to M$ be a principal $G$-bundle over a closed oriented $4$-manifold $M$, and let $A$ be a connection on $P$
with curvature $F_A\in \Omega^2(M;\operatorname{ad}P)$.
For a complex representation $\rho:G\to U(V)$, write $\Tr_\rho$ for the trace in $\operatorname{End}(V)$; then Chern--Weil theory gives
\begin{equation}\label{eq:TrF2_vs_basic_charge}
\frac{1}{8\pi^2}\int_M \Tr_\rho(F_A\wedge F_A)
\;=\;
\ell_\rho \, k(P),
\end{equation}
where $\ell_\rho$ is the Dynkin index of the representation $\rho$ , and where
\begin{equation}
k(P)\;:=\;\big\langle c^*y_4(G),[M]\big\rangle\in\mathbb Z
\end{equation}
is the intrinsic topological charge of $P$, defined using the basic generator $y_4(G)\in H^4(BG,\mathbb Z)$ and the
classifying map $c:M\to BG$ of $P$.
Thus, for a fixed bundle $P$, 
the value of the Chern--Weil integral depends on the representation $\rho$. The set of all possible values obtained by varying $\rho$ generates the ideal $d_G\mathbb Z$, where
\begin{equation*}
d_G=\gcd\{\ell_\rho:\rho\in R(G)\}.
\end{equation*}

\begin{rem}[Application: Green--Schwarz trace ratios]\label{rem:GS_trace_ratio}
Trace conversion factors in anomaly cancellation can be expressed in terms of Dynkin indices of $G$ and the dual Coxeter number.
If $\Tr_{\ad}$ and $\Tr_{\rho}$ denote traces in the adjoint and in a representation $\rho$, then
$$
\Tr_{\ad}(F^2) \;=\; A_{\ad/ \rho}\,\Tr_{\rho}(F^2),
\qquad\text{with}\qquad
A_{\adj/\rho} \;=\; \frac{\ell_{\ad}}{\ell_\rho}.
$$
Using $\ell_{\adj}=2h^\vee$ 
 (cf.\ \eqref{eq:ell-adj})
, one gets
\begin{equation}
A_{ad/\rho}=\frac{2h^\vee}{\ell_\rho}.
\end{equation}
In particular, if $\rho_F$ is the basic representation of $G$, one has  $\ell_{\rho_F}=d_G$, and therefore: 
\begin{equation}
A_{\rho_F}=\frac{2h^\vee}{d_G},
\end{equation}
recovering the standard factors in anomaly cancellations \cite{Green:1984sg,Erler, Schwarz95,Sagnotti}
\begin{align}
& \Tr_{\adj} (F^2)=2N\     \Tr_{\mathbf N} (F^2), &  SU(N), \\
& \Tr_{\adj} (F^2)=2(N+1)\   \Tr_{\mathbf{N}} (F^2),  & Sp(N), \\
& \Tr_{\adj} (F^2)=(N-2)\  \Tr_{\mathbf{N}}(F^2),   & Spin(N), \\
& \Tr_{\adj} (F^2)=4\ \Tr_{\mathbf{7}}(F^2), & G_2, \\
& \Tr_{\adj} (F^2)=3\ \Tr_{\mathbf{26}}(F^2), \quad &  F_4, \\
& \Tr_{\adj} (F^2)=4\ \Tr_{\mathbf{27}}(F^2), \quad &  E_6,\\
& \Tr_{\adj} (F^2)=3\  \Tr_{\mathbf{56}}(F^2), \quad & E_7,\\
& \Tr_{\adj} (F^2)=1\  \Tr_{\mathbf{248}}(F^2), \quad & E_8.
\end{align}
\end{rem}

%\clearpage
\begin{table}[H]
\begin{center}
\scalebox{.98}{
\renewcommand{\arraystretch}{1.2}
\begin{tabular}{|c|c|c|c|c|c|c|c|c|c|}
\hline
Lie algebra $\mathfrak{g}$
& A$_\ell$ ($\ell\ge 1$)
& B$_\ell$ ($\ell\ge 3$)
& C$_\ell$ ($\ell\ge 2$)
& D$_\ell$ ($\ell\ge 4$)
& G$_2$
& F$_4$
& E$_6$
& E$_7$
& E$_8$ \\
\hline
 Group $G$
& $SU(\ell+1)$ 
& $Spin(2\ell+1)$ 
& $Sp(\ell)$ 
& $Spin(2\ell)$ 
& G$_2$
& F$_4$
& E$_6$
& E$_7$
& E$_8$ \\
\hline
$\dim \mathfrak g$
& $\ell(\ell+2)$
& $\ell(2\ell+1)$
& $\ell(2\ell+1)$
& $\ell(2\ell-1)$
& $14$
& $52$
& $78$
& $133$
& $248$\\
\hline
$|\Phi|$ 
& $\ell(\ell+1)$
& $2\ell^2$
& $2\ell^2$
& $2\ell(\ell-1)$
& $12$
& $48$
& $72$
& $126$
& $240$\\
\hline
$\det C=|P/Q|$
& $\ell+1$
& $2$
& $2$
& $4$
& $1$
& $1$
& $3$
& $2$
& $1$\\
\hline
$h$
& $\ell+1$
& $2\ell$
& $2\ell$
& $2\ell-2$
& $6$
& $12$
& $12$
& $18$
& $30$\\
\hline
$h^\vee$
& $\ell+1$
& $2\ell-1$
& $\ell+1$
& $2\ell-2$
& $4$
& $9$
& $12$
& $18$
& $30$\\
\hline
$d_G$
& $1$
& $2$
& $1$
& $2$
& $2$
& $6$
& $6$
& $12$
& $60$\\
\hline
$\dim \rho_F$
& $\ell+1$
& $2\ell+1$
& $2\ell$
& $2 \ell$
& $7$
& $26$
& $27$
& $56$
& $248$\\
\hline
$A_{\rho_F}$
& $2(\ell+1)$
& $2\ell-1$
& $2(\ell+1)$
& $2\ell-2$
& $4$
& $3$
& $4$
& $3$
& $1$\\
\hline
\end{tabular}}
\end{center}
\caption{
Simple Lie algebra $\mathfrak{g}$, 
corresponding simply-connected compact Lie group $G$ with Lie algebra $\mathfrak{g}$, 
dimension, 
the number of roots $\Phi$, 
the determinant of the Cartan matrix  of $\mathfrak{g}$ (or equivalently the index of the root lattice ($R$) in the weight lattice ($P$)), 
Coxeter number $h$, dual Coxeter number $h^\vee$, 
the Dynkin index $d_G$ of the simply-connected compact Lie group $G$ , the integer $A_{\rho_F}$ in the trace identity ($\Tr_{adj} F^2=A_{\rho_F} \Tr_{\rho_F} F^2$) between the adjoint and the fundamental representation.  }
\label{Table:DataSimpleGroup}
\end{table}

\begin{table}[H]
\begin{center}
\scalebox{.9}{
\begin{tabular}{c}
{\scalebox{1.1}{\begin{tikzpicture}
\node  at (-1,0){A$_n$};
\node[draw,circle,thick,scale=1.1, label=right:{$\binom{n-1}{0}$}] (1) at (0,0){};
\node[draw,circle,thick,scale=1.1, label=right:{$\binom{n-1}{1}$}] (2) at (0,-1){};
\node[draw,circle,thick,scale=1.1, label=right:{$\binom{n-1}{k-1}$}] (3) at (0,-2.4){};
\node[draw,circle,thick,scale=1.1, label=right:{$\binom{n-1}{n-2}$}] (4) at (0,-3.8){};
\node[draw,circle,thick,scale=1.1, label=right:{$\binom{n-1}{n-1}$}] (5) at (0,-4.8){};
\draw[thick] (1)--(2);
\draw[thick,dashed] (2)--(3)--(4);
\draw[thick] (4)--(5);
\end{tikzpicture}}
}

{\scalebox{1.1}{\begin{tikzpicture}
\node  at (-1,0){B$_n$};
\node[draw,circle,thick,scale=1.1, label=right:{$2\binom{2n-1}{0}$}] (1) at (0,0){};
\node[draw,circle,thick,scale=1.1, label=right:{$2\binom{2n-1}{1}$}] (2) at (0,-1){};
\node[draw,circle,thick,scale=1.1, label=right:{$2\binom{2n-1}{k-1}$}] (3) at (0,-2.4){};
\node[draw,circle,thick,scale=1.1, label=right:{$2\binom{2n-1}{n-2}$}] (4) at (0,-3.8){};
\node[draw,circle,thick,scale=1.1, label=right:{$2^{n-2}$}] (5) at (0,-4.8){};
\draw[thick] (1)--(2);
\draw[thick,dashed] (2)--(3)--(4);
\draw[thick, double, double distance =3pt] (4)--(5);
		\draw[thick] plot [smooth,tension=1] coordinates { (0,-4.5) (.13,-4.2) (.27,-4.1)};
			\draw[thick] plot [smooth,tension=1] coordinates { (0,-4.5) (-.13,-4.2) (-.27,-4.1)};				

\end{tikzpicture}}
}

{\scalebox{1.1}{\begin{tikzpicture}
\node  at (-1,0){C$_n$};
\node[draw,circle,thick,scale=1, label=right:{$\binom{2n-1}{0}$}] (1) at (0,0){};
\node[draw,circle,thick,scale=1, label=right:{$\binom{2n-1}{1}-\binom{2n-1}{0}$}] (2) at (0,-1){};
\node[draw,circle,thick,scale=1, label=right:{$\binom{2n-1}{k-1}-\binom{2n-1}{k-2}$}] (3) at (0,-2.4){};
\node[draw,circle,thick,scale=1, label=right:{$\binom{2n-1}{n-2}-\binom{2n-1}{n-3}$}] (4) at (0,-3.8){};
\node[draw,circle,thick,scale=1, label=right:{$\binom{2n-1}{n-1}-\binom{2n-1}{n-2}$}] (5) at (0,-4.8){};
\draw[thick] (1)--(2);
\draw[thick,dashed] (2)--(3)--(4);
\draw[thick, double, double distance=3pt] (4)--(5);

\draw[thick] plot [smooth,tension=1] coordinates { (0,-4.1) (.13,-4.4) (.27,-4.5)};
			\draw[thick] plot [smooth,tension=1] coordinates { (0,-4.1) (-.13,-4.4) (-.27,-4.5)};	

\end{tikzpicture}}
}
{\scalebox{1.1}{\begin{tikzpicture}
\node  at (-1,0){D$_n$};
\node[draw,circle,thick,scale=1, label=right:{$2\binom{2n-2}{0}$}] (1) at (0,0){};
\node[draw,circle,thick,scale=1, label=right:{$2\binom{2n-2}{1}$}] (2) at (0,-1){};
\node[draw,circle,thick,scale=1, label=right:{$2\binom{2n-2}{k-1}$}] (3) at (0,-2.4){};
\node[draw,circle,thick,scale=1, label=right:{$2\binom{2n-2}{n-3}$}] (4) at (0,-3.8){};
\node[draw,circle,thick,scale=1, label= left:{$2^{n-3}$}] (5) at (-.5,-4.8){};
\node[draw,circle,thick,scale=1, label= right:{$2^{n-3}$}] (6) at (.5,-4.8){};
\draw[thick] (1)--(2);
\draw[thick,dashed] (2)--(3)--(4);
\draw[thick] (4)--(5);
\draw[thick] (4)--(6);
\end{tikzpicture}}
}
\\
\\
{\scalebox{1.1}{\begin{tikzpicture}
\node  at (-1,-5){G$_2$};
\node[draw,circle,thick,scale=1, label= right:{$2$}] (5) at (0,-5){};
\node[draw,circle,thick,scale=1, label= right:{$8$}] (6) at (0,-6){};
\draw[thick, double, double distance =5 pt] (5)--(6);
\draw[thick] (5)--(6);
			\draw[thick] plot [smooth,tension=1] coordinates { (0,-5.3) (.13,-5.6) (.27,-5.7)};
			\draw[thick] plot [smooth,tension=1] coordinates { (0,-5.3) (-.13,-5.6) (-.27,-5.7)};	
\end{tikzpicture}}}
{\scalebox{1}{\begin{tikzpicture}
\node  at (-1,0){F$_4$};
\node[draw,circle,thick,scale=1, label= right:{$6$}] (1) at (0,0){};
\node[draw,circle,thick,scale=1, label= right:{$126$}] (2) at (0,-1){};
\node[draw,circle,thick,scale=1, label= right:{$882$}] (3) at (0,-2){};
\node[draw,circle,thick,scale=1, label= right:{$18$}] (4) at (0,-3){};
\draw[thick] (1)--(2);
\draw[thick] (3)--(4);
\draw[thick, double, double distance=3pt] (2)--(3);
\draw[thick] plot [smooth,tension=1] coordinates { (0,-1.3) (.13,-1.6) (.27,-1.7)};
			\draw[thick] plot [smooth,tension=1] coordinates { (0,-1.3) (-.13,-1.6) (-.27,-1.7)};
\end{tikzpicture}}}

{\scalebox{1}{\begin{tikzpicture}
\node  at (-1,0){E$_6$};
\node[draw,circle,thick,scale=1, label= left:{$24$}] (0) at (-1,-2){};
\node[draw,circle,thick,scale=1, label= right:{$6$}] (1) at (0,0){};
\node[draw,circle,thick,scale=1, label= right:{$150$}] (2) at (0,-1){};
\node[draw,circle,thick,scale=1, label= right:{$1800$}] (3) at (0,-2){};
\node[draw,circle,thick,scale=1, label= right:{$150$}] (4) at (0,-3){};
\node[draw,circle,thick,scale=1, label= right:{$6$}] (5) at (0,-4){};
\draw[thick] (1)--(2)--(3)--(4)--(5);
\draw[thick] (0)--(3);
\end{tikzpicture} }}
{\scalebox{1}{\begin{tikzpicture}
\node  at (-1,0){E$_7$};
\node[draw,circle,thick,scale=1, label= left:{$360$}] (0) at (-1,-3){};
\node[draw,circle,thick,scale=1, label= right:{$36$}] (1) at (0,-5){};
\node[draw,circle,thick,scale=1, label= right:{$4680$}] (2) at (0,-4){};
\node[draw,circle,thick,scale=1, label= right:{$297000$}] (3) at (0,-3){};
\node[draw,circle,thick,scale=1, label= right:{$17160$}] (4) at (0,-2){};
\node[draw,circle,thick,scale=1, label= right:{$648$}] (5) at (0,-1){};
\node[draw,circle,thick,scale=1, label= right:{$12$}] (6) at (0,0){};
\draw[thick] (1)--(2)--(3)--(4)--(5)--(6);
\draw[thick] (0)--(3);
\end{tikzpicture}}}

{\scalebox{1}{\begin{tikzpicture}
\node  at (-1,0){E$_8$};
\node[draw,circle,thick,scale=1, label= left:{$85500$}] (0) at (-1,-4){};
\node[draw,circle,thick,scale=1, label= right:{$60$}] (1) at (0,0){};
\node[draw,circle,thick,scale=1, label= right:{$14700$}] (2) at (0,-1){};
\node[draw,circle,thick,scale=1, label= right:{$177840$}] (3) at (0,-2){};
\node[draw,circle,thick,scale=1, label= right:{$141605100$}] (4) at (0,-3){};
\node[draw,circle,thick,scale=1, label= right:{$8345660400$}] (5) at (0,-4){};
\node[draw,circle,thick,scale=1, label= right:{$5292000$}] (6) at (0,-5){};
\node[draw,circle,thick,scale=1, label= right:{$1500$}] (7) at (0,-6){};
\draw[thick] (1)--(2)--(3)--(4)--(5)--(6)--(7);
\draw[thick] (0)--(5);
\end{tikzpicture}}}
\end{tabular}}
\end{center}
\caption{Dynkin indices of the fundamental representations of simple Lie algebras $A_n$ ($n\geq 1$), $B_n$ ($n\geq 3$), $C_n$ ($n\geq 2$), $D_n$ ($n\geq 4$), $G_2$, $F_4$, $E_6$, $E_7$, and $E_8$. 
This table is obtained from  \cite[Table 5]{Dynkin.SubA} after fixing some inaccuracies  for B$_n$,  C$_n$, and E$_8$.  }\label{DynkinT5}
\end{table}
%\newpage

\section{Index-one versus index >1 in instanton physics}\label{sec:instanton_story}

Throughout this section $H$ denotes a compact, connected, simply-connected, simple
Lie group (the UV gauge group), and $G\subset H$ a connected closed subgroup
(the unbroken gauge group after spontaneous symmetry breaking).
We write $f:G\hookrightarrow H$ for the inclusion and $j_f$ for its Dynkin
embedding index in the normalization fixed in Section~\ref{sec:DI} and used in
Proposal~\ref{prop:TopologicalConservation}.

\subsection{Instantons and characteristic classes}

A basic message of Atiyah--Hitchin--Singer \cite{AHS} is that self-dual
(anti-self-dual) Yang--Mills solutions in four dimensions naturally decompose into
topological sectors determined by the underlying principal bundle.
On $S^{4}$ (or on $\mathbb R^{4}$ with a framing at infinity), principal
$K$-bundles are classified by
\begin{equation*}
[S^{4},BK]\;\cong\;\pi_{4}(BK)\;\cong\;\pi_{3}(K),
\end{equation*}
and for $K$ simple and simply connected one has $\pi_{3}(K)\cong\mathbb Z$.
This integer is the instanton number \cite{Bernard,AHS,Weinberg}: for unitary groups it is the second Chern
class evaluated on $[S^{4}]$, and more generally it is represented (after a fixed
normalization of the invariant bilinear form) by the corresponding integral class
in $H^{4}(BK,\mathbb Z)$ via Chern--Weil theory.

Once instanton number is understood as the generator of $\pi_{3}(K)$ (equivalently
$H^{4}(BK,\mathbb Z)$), embeddings of structure groups become rigidly constrained.
Given $f:G\hookrightarrow H$, extension of structure group induces
\begin{equation*}
f_*:\pi_{3}(G)\longrightarrow \pi_{3}(H),
\qquad
(Bf)^*:H^{4}(BH,\mathbb Z)\longrightarrow H^{4}(BG,\mathbb Z),
\end{equation*}
and, after compatible choices of generators, both maps are multiplication by the
same integer $j_f$.
Thus $j_f=1$ is exactly the condition that \emph{instanton charge is preserved}
under extension of structure group.
In AHS this plays a concrete structural role: index-one subgroup chains provide a
mechanism for transporting known self-dual solutions (starting from $SU(2)$)
to larger gauge groups without altering their topological charge, enabling
dimension-count arguments on moduli spaces to force the existence of irreducible
solutions beyond those obtained by reduction to proper subgroups \cite{AHS}.

\subsection{Intriligator--Seiberg: when ``broken-sector instantons'' are physically meaningful}

Intriligator and Seiberg highlight a sharp physical distinction between $j_f=1$
and $j_f>1$ in the Higgs phase of a gauge theory with UV group $H$ broken to an
unbroken subgroup $G\subset H$ \cite[Footnote~2]{Intriligator:1995id}.
The point may be expressed directly in the language of instanton number:
if an $H$-bundle reduces to $G$ (equivalently, the corresponding connection is
supported in the unbroken subgroup), then its instanton number lies in the subgroup
$f_*\pi_{3}(G)\subset \pi_{3}(H)$, hence is divisible by $j_f$ after identifying
both $\pi_3$'s with $\mathbb Z$.
Consequently:
\begin{itemize}
\item If $j_f=1$ and $G$ is non-Abelian, the basic unit-charge $H$-instanton
sector can be represented inside $G$.  In this case it is neither gauge-invariant
nor necessary to isolate a separate notion of ``instantons in the broken part'' in
the low-energy description: the unbroken $G$-theory already accesses the minimal
instanton sector.
\item If $j_f>1$, then a unit $H$-instanton cannot be represented by any
configuration supported entirely inside $G$, since every $G$-supported sector has
charge in $j_f\mathbb Z\subset\mathbb Z\cong\pi_3(H)$.
Thus there exist genuine UV topological sectors whose representatives necessarily
involve the massive fields in the broken directions $\mathfrak h/\mathfrak g$.
When these massive multiplets are integrated out, their semiclassical contributions
(often described as \emph{constrained instantons}) must be retained, because they are
not reproduced by the nonperturbative dynamics of the low-energy $G$-theory alone.
\end{itemize}
This is precisely the mechanism by which $j_f>1$ signals a failure of ``charge
minimality'' to survive symmetry breaking: the UV theory possesses sectors invisible
to the unbroken subgroup.

\subsection{Third homotopy group of H/G and the Dynkin embedding index}

Cs\'aki and Murayama develop the Intriligator--Seiberg perspective systematically by identifying the
homogeneous space $H/G$ as the natural topological object controlling instanton
effects in the Higgs phase \cite{Csaki}.
Since $G$ is closed in $H$, the projection $H\to H/G$ is a principal
$G$-bundle, and the long exact sequence in homotopy yields a canonical quotient
description
\begin{equation*}
\pi_{3}(H/G)\;\cong\;\pi_{3}(H)\big/ f_*\pi_{3}(G).
\end{equation*}
When $H$ and $G$ are simple and simply connected, $\pi_{3}(H)\cong\pi_{3}(G)\cong
\mathbb Z$, and the subgroup $f_*\pi_{3}(G)\subset\pi_{3}(H)$ has index $j_f$.
Hence one obtains the discrete remnant
\begin{equation*}
\pi_{3}(H/G)\;\cong\;\mathbb Z/j_f\mathbb Z,
\end{equation*}
which we prove in Proposition~\ref{prop:pi3_H_over_G}.
In particular, $j_f>1$ produces additional $\mathbb Z/j_f\mathbb Z$-valued
coset-instanton sectors, exactly corresponding to the UV instanton sectors that cannot
be deformed into configurations supported entirely inside the unbroken subgroup $G$,
while $j_f=1$ implies $\pi_{3}(H/G)=0$ and no such extra sectors exist.

Across AHS, Intriligator--Seiberg, and Cs\'aki--Murayama, the same integer $j_f$
emerges as the discriminator between two regimes:
\begin{align*}
\begin{aligned}
    j_f=1 \ \Longleftrightarrow\ \text{preservation of the minimal quantized topological sector,}\\
    j_f>1 \ \Longleftrightarrow\ \text{new coset sectors and charge fractionalization phenomena.}
\end{aligned}
\end{align*}
This provides a concrete gauge-theoretic interpretation of the index-one condition:
it is the precise criterion for the conservation of the basic $\pi_{3}$-classified
instanton charge (and, by Proposal~\ref{prop:TopologicalConservation}, the allied
quantization levels) under embeddings of gauge groups.

\subsection{Dynamical scaling: coupling matching, instanton actions, and RG scales}
\label{subsec:dynamical_scaling}

Let $H$ be a compact, simple gauge group spontaneously broken to a connected subgroup
$f:G\hookrightarrow H$ at a mass scale $M$ by a Higgs expectation value.
Write $\mathfrak h=Lie(H)$ and $\mathfrak g=Lie(G)$, and denote by
\begin{equation}
f_*:\mathfrak g\longrightarrow \mathfrak h
\end{equation}
the induced Lie algebra homomorphism (the differential of $f$ at the identity).
We write $(\cdot,\cdot)_H$ and $(\cdot,\cdot)_G$ for the \emph{basic} $\Ad$-invariant inner products
on $\mathfrak h$ and $\mathfrak g$ (normalized as in Section~\ref{sec:DI}),
and we denote by $j_f\in\mathbb Z_{>0}$ the Dynkin embedding index, i.e.
\begin{equation}\label{eq:dynkin_index_embedding_inner_product}
\bigl(f_*(X),f_*(Y)\bigr)_H \;=\; j_f\,(X,Y)_G \qquad (X,Y\in\mathfrak g).
\end{equation}
This topological normalization has immediate dynamical consequences, because the Yang--Mills kinetic term
is defined using the same invariant quadratic form.

\paragraph{Tree-level matching of gauge couplings.}
In the ultraviolet $H$-theory, take the kinetic term normalized by the basic inner product:
\begin{equation}\label{eq:YM_H_basic}
\mathcal L_{H}\supset -\frac{1}{4g_H^2}\,(F_H,F_H)_H,
\end{equation}
and similarly in the infrared $G$-theory,
\begin{equation}\label{eq:YM_G_basic}
\mathcal L_{G}\supset -\frac{1}{4g_G^2}\,(F_G,F_G)_G.
\end{equation}
Just below $M$, the surviving field strength takes values in $f_*(\mathfrak g)\subset\mathfrak h$.
Restricting \eqref{eq:YM_H_basic} along $f_*$ and using \eqref{eq:dynkin_index_embedding_inner_product} gives
\begin{equation}
-\frac{1}{4g_H^2}\,(F_G,F_G)_H \;=\; -\frac{j_f}{4g_H^2}\,(F_G,F_G)_G.
\end{equation}
Comparing with \eqref{eq:YM_G_basic} yields the tree-level matching condition
\begin{equation}\label{eq:coupling_matching_basic}
\frac{1}{g_G^2(M)} \;=\; \frac{j_f}{g_H^2(M)} \,+\, \Delta_{\mathrm{thr}}(M),
\end{equation}
where $\Delta_{\mathrm{thr}}(M)$ denotes threshold corrections from integrating out the heavy vector multiplet
(or heavy gauge bosons) in the coset $H/G$ and any additional massive matter at the breaking scale.

\paragraph{Instanton action and ``fractionalization'' in the low-energy description.}
For a configuration whose curvature is supported in the embedded subalgebra $f_*(\mathfrak g)\subset\mathfrak h$,
the Yang--Mills action computed using the $H$-normalization differs from the $G$-normalization by exactly the same
factor $j_f$.  In particular, for an embedded $G$-instanton with integer charge
$k_G\in\pi_3(G)\cong\mathbb Z$ (equivalently, the class detected by the basic generator $y_4(G)\in H^4(BG,\mathbb Z)$),
the corresponding $H$-charge is
\begin{equation}\label{eq:charge_scaling}
k_H \;=\; j_f\,k_G,
\end{equation}
reflecting the universal scaling of the basic class under $Bf$ proved in Section~\ref{sec:universalscaling}.
Consequently, the classical instanton actions satisfy
\begin{equation}
S_H(k_H)\;=\;\frac{8\pi^2}{g_H^2}\,k_H,
\qquad
S_G(k_G)\;=\;\frac{8\pi^2}{g_G^2}\,k_G,
\end{equation}
and \eqref{eq:coupling_matching_basic} implies (at tree level, ignoring thresholds) that embedded configurations match:
\begin{equation}
S_G(k_G)\;=\;S_H(k_H)\qquad\text{for configurations supported in }f_*(\mathfrak g).
\end{equation}
Thus $j_f$ does \emph{not} introduce a discrepancy of actions for an honest embedded $G$-instanton; it introduces a
discrepancy between what the two theories call a ``unit'' of topological charge.

The genuinely new phenomenon for $j_f>1$ concerns configurations of the full $H$-theory that cannot be represented as
honest $G$-instantons after symmetry breaking.
A unit $H$-instanton has $k_H=1$, and by \eqref{eq:charge_scaling} it would correspond formally to $k_G=1/j_f$ in the
$G$-normalization. Since $\pi_3(G)\cong\mathbb Z$, there is no smooth $G$-bundle of charge $1/j_f$; rather, the missing
contributions appear as \emph{constrained instantons} (often called \emph{coset instantons}) whose core explores the
broken directions and whose size is stabilized by the Higgs scale.
This mechanism and its consequences in partially broken gauge groups are analyzed in detail by Cs\'aki and
Murayama~\cite{Csaki}.

\subsection{Parallel from gauge theory: why index-1 embeddings are structurally useful}
\label{subsec:gauge_theory_parallel}

It is worth noting that the practical usefulness of index-one embeddings is already visible in the early
gauge-theory literature: one efficient route to constructing instantons for a general compact simple group
$H$ is to start from the BPST $SU(2)$ instanton and extend the structure group along a chosen embedding
\begin{equation}
i:SU(2)\hookrightarrow H.
\end{equation}
When $i$ has Dynkin embedding index $j_i=1$, the resulting configuration lies in the \emph{minimal}
topological sector of $H$ (i.e.\ it represents the generator of $\pi_3(H)\cong\mathbb Z$), so it provides a
genuine ``seed'' for the full moduli problem without an \emph{a priori} charge inflation.

\medskip
\noindent\textbf{Highest-root $SU(2)$ as the canonical charge-one seed.}
Both Atiyah--Hitchin--Singer and Bernard--Christ--Guth--Weinberg single out the $SU(2)$ subgroup associated to
a \emph{long root}; with our normalization this is equivalently the \emph{highest-root} subgroup
\begin{equation}
i_\theta:SU(2)_\theta\hookrightarrow H,
\end{equation}
whose induced class $[S^3_\theta]= (i_\theta)_*[SU(2)]$ generates $\pi_3(H)$ and hence defines the unit-charge
sector. This is precisely the normalization adopted throughout this paper (Section~\ref{sec:DI}).

\medskip
\noindent\textbf{Atiyah--Hitchin--Singer: index-one chains for charge bookkeeping.}
AHS emphasize that one can transport existence results and compare reducible strata along subgroup chains while
keeping the notion of ``unit instanton'' stable. Concretely, they use familiar chains of inclusions
in which each step has Dynkin index $1$ (so the induced map on $\pi_3$ is multiplication by $1$):
\begin{align}\label{eq:AHS_classical_chains}
 SU(n) &\hookrightarrow SU(n+1),\\
 Spin(n) &\hookrightarrow Spin(n+1), \\
 Sp(n) &\hookrightarrow Sp(n+1), \\
 G_2 &\hookrightarrow Spin(7), \\
 Spin(9) &\hookrightarrow F_4 \hookrightarrow E_6 \hookrightarrow E_7 \hookrightarrow E_8,
\end{align}
In this way, ``charge one'' remains a stable concept across the families and exceptional series, which is
structurally useful when analyzing how reducible instantons sit inside the full moduli space \cite{AHS}.

\medskip
\noindent\textbf{Bernard--Christ--Guth--Weinberg: embedding BPST into the highest-root subgroup.}
Bernard--Christ--Guth--Weinberg operationalize the same idea by embedding the BPST $SU(2)$ solution into an
$SU(2)\subset H$ chosen to minimize the resulting topological charge; in a simple group this minimizing choice
is achieved by taking an $SU(2)$ corresponding to a \emph{maximal-length} (i.e.\ long) root, hence one may take
the highest root $\theta$. This produces the $k=1$ sector for an arbitrary gauge group $H$ and provides the
basic building block for semiclassical analysis in the minimal sector \cite{Bernard}.

\medskip
\noindent\textbf{Takeaway.}
Index-one embeddings are therefore not merely a normalization convenience: they are a practical device that lets
one seed the moduli problem for a general group in the minimal sector, and then compare reducible constructions
coming from proper subgroups against the expected dimension to force the existence of irreducible solutions.
This is exactly the same role played by $j_f=1$ throughout our discussion of conservation vs.\ fractionalization
of instanton sectors under symmetry breaking.

\paragraph{One-loop scale matching.}
At one loop, in a mass-independent scheme and away from thresholds, the running couplings satisfy
\begin{equation}
\frac{8\pi^2}{g_H^2(\mu)} = b_H\log\!\Bigl(\frac{\mu}{\Lambda_H}\Bigr),
\qquad
\frac{8\pi^2}{g_G^2(\mu)} = b_G\log\!\Bigl(\frac{\mu}{\Lambda_G}\Bigr),
\end{equation}
where $b_H$ and $b_G$ are the one-loop beta-function coefficients in the UV and IR effective theories.
Inserting \eqref{eq:coupling_matching_basic} at $\mu=M$ gives
\begin{equation}
b_G\log\!\Bigl(\frac{M}{\Lambda_G}\Bigr)
=
j_f\,b_H\log\!\Bigl(\frac{M}{\Lambda_H}\Bigr)
+
8\pi^2\,\Delta_{\mathrm{thr}}(M),
\end{equation}
and exponentiating this identity yields the scale relation
\begin{equation}\label{eq:Lambda_matching_basic}
\Bigl(\frac{\Lambda_G}{M}\Bigr)^{b_G}
\;=\;
\Bigl(\frac{\Lambda_H}{M}\Bigr)^{j_f\,b_H}
\times \exp\!\Bigl(-8\pi^2\,\Delta_{\mathrm{thr}}(M)\Bigr).
\end{equation}
Thus index-one embeddings ($j_f=1$) behave as a seamless identification (up to thresholds), whereas $j_f>1$ produces a
systematic power-law distortion governed by the same integer that controls the topological scaling of $H^4(B-)$ and
$\pi_3(-)$.

\paragraph{Remarks on conventions and scope.}
\begin{itemize}
\item The matching \eqref{eq:coupling_matching_basic} is stated in the normalization of the basic inner products used
throughout this paper. If one instead normalizes the kinetic term using a representation trace $\Tr_\rho$, the same
discussion holds after inserting the Dynkin index $\ell_\rho$ relating $\Tr_\rho$ to the basic form.
\item The phrase ``fractional $G$-charge'' is an \emph{effective} low-energy description of relevant $H$-theory saddles;
it does not assert the existence of honest $G$-bundles with non-integer characteristic class.
\item In supersymmetric settings (notably $\mathcal N=1$ or $\mathcal N=2$), it is often cleaner to state matching in
terms of holomorphic couplings and properly defined threshold functions; we do not pursue those refinements here.
\end{itemize}

\medskip

\subsection{Connection to F-theory and string compactifications.}
In F-theory compactifications with 7-branes wrapping a divisor $S$, the four-dimensional gauge kinetic term
descends from the brane worldvolume action and is controlled (at leading order) by the complexified K\"ahler modulus
$T_S$, so that $\mathrm{Re}(f)\sim \mathrm{Re}(T_S)\propto \mathrm{Vol}(S)$ up to threshold effects and
flux-dependent corrections. Equivalently, once one fixes an $\Ad$-invariant quadratic form on the gauge algebra, the
overall normalization of the gauge kinetic term is fixed; this is precisely the choice encoded by our \emph{basic}
normalization from Section~\ref{sec:DI}. \cite{BeasleyHeckmanVafa,WeigandTASI}

If a gauge group $H$ is Higgsed to an embedded subgroup $f\colon G\hookrightarrow H$ (e.g.\ in a Higgsing within a
fixed brane stack on $S$), then restricting the basic quadratic form from $\mathfrak h$ to
$f_*(\mathfrak g)\subset \mathfrak h$ introduces the Dynkin embedding index:
\begin{equation}
(f_*(X),f_*(Y))_H = j_f\,(X,Y)_G.
\end{equation}
Consequently, when the infrared coupling $g_G$ is expressed in the \emph{basic} normalization for $G$, it matches the
ultraviolet coupling $g_H$ by the same factor as in \eqref{eq:coupling_matching_basic},
\begin{equation}
\frac{1}{g_G^2(M)} = \frac{j_f}{g_H^2(M)} + \Delta_{\mathrm{thr}}(M),
\end{equation}
so the integer $j_f$ should be read as a \emph{normalization conversion} between the $H$- and $G$-basic quadratic forms
rather than as a literal rescaling of the geometric volume. \cite{DonagiWijnholt,WeigandTASI}

In heterotic dual descriptions, analogous normalization data are often packaged by affine (Kac--Moody) level(s) $k_a$,
which multiply the basic quadratic form and appear as overall factors in gauge couplings; in this sense, $j_f$ plays the
same bookkeeping role for subgroup embeddings at the level of quadratic-form normalization. \cite{MorrisonVafaII,DienesLevels}

Finally, the quotient $\pi_3(H/G)\cong \mathbb Z/j_f\mathbb Z$ identifies additional topological sectors in the Higgs
phase (the ``coset'' sectors of \cite{Csaki}), i.e.\ UV saddles that cannot be represented by configurations supported
entirely in the unbroken subgroup $G$. In F-theory, nonperturbative effects are naturally sourced by Euclidean D3-branes
(E3-instantons) or their F-theory/M-theory lifts, and gauge-instanton effects on 7-branes admit such Euclidean-brane
descriptions in appropriate limits. \cite{CveticThreeLooks}
It is therefore natural to expect that index-$j_f$ coset sectors can manifest as discrete selection rules or additional
nonperturbative contributions in symmetry-breaking backgrounds, although the precise realization depends on the global
compactification data (flux quantization, global form of the gauge group, orientifold/F-theory uplift, and thresholds).

\section{Index-one embeddings in F-theory: a genericity heuristic}\label{sec:Ftheory_genericity}

This section is a heuristic bridge between the \emph{topological} role of the Dynkin embedding index $j_f$
and the \emph{geometric} engineering of gauge algebras in F-theory.
Recall that $j_f$ is defined by comparing the basic invariant inner products under an embedding
$f\colon G\hookrightarrow H$,
\begin{equation}
(f_*(X),f_*(Y))_H \;=\; j_f\,(X,Y)_G,
\end{equation}
and that the same integer controls the induced maps on $H^4(B-)$ and on $\pi_3(-)$, hence the scaling of instanton
charge and the appearance of coset sectors $\pi_3(H/G)\cong \mathbb Z/j_f\mathbb Z$ in the Higgs phase.

A recurring empirical feature of many \emph{standard} F-theory constructions (notably those built from Tate/Weierstrass
models with ``minimal'' codimension-two enhancements) is that the relevant subgroup embeddings often have $j_f=1$.
One should not interpret this as a theorem about all embeddings that can occur in string compactifications; rather, it
is a \emph{genericity heuristic}: the geometry typically realizes embeddings that are compatible with the most economical
lattice inclusion dictated by the singular fiber and its degenerations, and these frequently come with index~$1$.

\subsection{From Kodaira enhancements to Lie embeddings}

In F-theory, a gauge algebra $\mathfrak g$ is engineered along a divisor $\Sigma$ by specifying a Kodaira fiber type
over $\Sigma$ (in a given Weierstrass/Tate model). Along codimension-two loci in the base (matter curves), the fiber
degenerates further and the singularity type enhances to some $\mathfrak h\supset \mathfrak g$.
Algebraically, this produces an embedding (well-defined up to conjugacy)
\begin{equation}
\mathfrak g \hookrightarrow \mathfrak h,
\end{equation}
and a branching of the adjoint of $\mathfrak h$ under $\mathfrak g$ that controls the localized matter content.
A key point for our purposes is that \emph{there can be multiple non-conjugate embeddings} of the same abstract
Lie algebra $\mathfrak g$ into $\mathfrak h$, with \emph{different} Dynkin indices and different branching rules.
The geometry (monodromy, splitting/non-splitting of fibers, and global consistency conditions) selects which embedding is
actually realized.

\subsection{Regular vs.\ special embeddings, and what geometry tends to prefer}

Dynkin’s classification distinguishes:
\begin{itemize}
\item \emph{Regular} (or $R$-) subalgebras: those obtained by taking a root subsystem (equivalently, by deleting nodes from
the extended Dynkin diagram in the simply-laced setting).
\item \emph{Special} (or $S$-) subalgebras: those not arising from a root subsystem, e.g.\ embeddings defined by
distinguished representations (principal $\mathfrak{sl}_2$, triality-related embeddings, etc.).
\end{itemize}
In many Tate-model engineering situations, codimension-two enhancements are ``minimal'' in the sense that they are
obtained by imposing one additional vanishing condition on a coefficient, and the resulting inclusion
$\mathfrak g\subset \mathfrak h$ is often regular (or a close cousin obtained by monodromy/folding in the non-simply-laced
cases). This is one reason $j_f=1$ appears frequently in standard model-building examples.

However, \textbf{one should be careful:}
\begin{center}
\emph{``regular'' does not mean ``index one'', and ``index one'' does not mean ``regular''.}
\end{center}
In particular, there are important \emph{special} subalgebras of Dynkin index~$1$, so ``$j_f=1$'' is \emph{not} a proxy
for ``$R$-type''.

\subsection{Index one is common, but not synonymous with regular: lessons from explicit F-theory analyses}

A useful guide is the systematic analysis of semisimple subalgebras and their Dynkin indices in the context of
F-theory matter models (and the associated ``characteristic representations''); see, e.g., \cite{EK11}.
Two concrete lessons are:

\smallskip
\noindent\textbf{(1) Geometry often selects the index-one embedding among several algebraic possibilities.}
For example, for $\mathfrak g_2\subset \mathfrak e_6$ there is an $R$-subalgebra embedding of Dynkin index~$1$ and a
maximal $S$-subalgebra embedding of Dynkin index~$3$; the corresponding characteristic representations differ markedly
(heuristically: the index-one embedding reproduces the ``expected'' matter in the standard geometric model, while the
higher-index embedding does not). Ref.~\cite{EK11} emphasizes precisely this kind of selection mechanism.

\smallskip
\noindent\textbf{(2) There exist $S$-subalgebras of Dynkin index~$1$.}
A striking example is $\mathfrak b_3\subset \mathfrak d_4$, where $\mathfrak b_3$ occurs as (triality-related) maximal
$S$-subalgebras \emph{of Dynkin index one}.
Thus, even if one expects ``generic'' Tate enhancements to correlate strongly with ``regular'' inclusions, the index-one
condition itself is \emph{broader} and can occur in genuinely special ways.

\smallskip
Taken together, these points support the following refined heuristic:

\begin{quote}
\emph{F-theory constructions built from minimal Weierstrass/Tate tunings frequently realize embeddings of smallest
possible Dynkin index, and in many common families this smallest index is~$1$; nonetheless, index-one embeddings may be
regular or special, and higher-index embeddings can arise but typically require additional structure beyond the most
generic enhancement pattern.}
\end{quote}

\subsection{Where j>1 can plausibly enter: non-generic corners and additional structure}

From this perspective, $j_f>1$ should not be read as ``forbidden'' in F-theory; rather, it is a signal that the geometry
is realizing a \emph{non-minimal normalization of the embedded quadratic form} (equivalently, a non-primitive inclusion at
the level relevant to $H^4(B-)$ and $\pi_3$). Practically, this tends to correlate with one (or several) of the following
features:
\begin{itemize}
\item \emph{Non-generic complex-structure tunings} beyond the standard Tate ansatz (higher-codimension conditions, extra
factorizations, or non-transverse collision patterns).
\item \emph{Monodromy/flux constraints} that alter which roots survive globally and thereby change the realized embedding
class (and sometimes the global form of the gauge group).
\item \emph{Dual heterotic interpretations at higher affine level.} In heterotic duals, overall normalization data are
packaged by affine levels; engineering levels $>1$ is typically more constrained than level~$1$, and it is natural to
expect the F-theory avatar of such data to be correspondingly non-generic.
\end{itemize}

In such settings, our universal scaling results predict the \emph{topological signature} of these corners: the induced map
on $\pi_3$ (and on $H^4(B-)$) is multiplication by $j_f$, and the Higgs phase supports additional coset sectors
$\pi_3(H/G)\cong \mathbb Z/j_f\mathbb Z$, hence potentially additional constrained-instanton contributions (in the sense of \cite{Csaki}).

The slogan ``F-theory favors $j_f=1$'' is best understood as a \emph{genericity heuristic} about minimal enhancements and
minimal normalizations, not as a rigid dichotomy between $R$- and $S$-subalgebras.
Index-one embeddings can be regular or special, and explicit F-theory analyses show both phenomena; higher-index embeddings
are plausible but typically correlate with extra geometric or global input beyond the most generic Tate enhancement
patterns.

\section{K-theory and the Chern character: the topological bridge}\label{sec:K}

In this section we introduce the $K$-theoretic and cohomological framework that connects the
purely algebraic Dynkin indices ($j_f$, $\ell_\rho$) from Section~\ref{sec:DI} to the basic
topological groups $H^3(G,\mathbb Z)$ and $H^4(BG,\mathbb Z)$ (Appendix~\ref{sec:A}).
Our goal is to explain how Dynkin indices appear as \emph{topological characteristic numbers}
via the odd Chern character (Harris), and equivalently via the suspended degree--$4$ Chern character (Naylor).

\begin{rem}[Physical interpretation]
From a physical perspective, the same basic invariant quadratic form on $\mathfrak g$ yields:
the instanton class on $BG$, the WZW $3$-form on $G$, and the suspended class on $\Sigma G$
detected by $\operatorname{ch}_2$ in Naylor's description of $K^{-1}(G)$.
\end{rem}

\subsection{Basic definitions in K-theory}

Let $X$ be a space with the homotopy type of a CW-complex.
The group $K^0(X)$ is the Grothendieck group of complex vector bundles on $X$, and
$\widetilde K^0(X)$ denotes reduced $K$-theory.

Odd $K$-theory admits a homotopy-theoretic description \cite{Atiyah,AH}:
\begin{equation}\label{eq:Kminus1_maps}
K^{-1}(X)\;\cong\;[X,U],
\end{equation}
where $U=\varinjlim U(n)$ is the stable unitary group.
By Bott periodicity there is a natural identification
\begin{equation}\label{eq:Bott_susp_identification}
K^{-1}(X)\;\cong\;\widetilde K^0(\Sigma X),
\end{equation}
and more generally $K^{-n}(X):=\widetilde K^0(\Sigma^n X)$ for $n\ge 1$.
Thus complex topological $K$-theory is $\mathbb Z/2\mathbb Z$-graded:
\begin{equation*}
K(X)=K^0(X)\oplus K^{-1}(X),\qquad K^{i}(X)\cong K^{i-2}(X).
\end{equation*}

\begin{notation}\label{not:K_pullback}
For a continuous map $g:X\to Y$, we write
\begin{equation*}
g^!:\ K^*(Y)\longrightarrow K^*(X)
\end{equation*}
for the induced pullback in (complex topological) $K$-theory.
This is the usual contravariant functoriality of $K$-theory (not a Gysin map).
\end{notation}

The \textbf{Chern character} is a natural transformation
\begin{equation*}
\operatorname{ch}:K^*(X)\longrightarrow H^*(X;\mathbb Q)
\end{equation*}
which becomes an isomorphism after tensoring with $\mathbb Q$.
We write $\operatorname{ch}_3$ for the degree--$3$ component on $K^{-1}(X)$, and
$\operatorname{ch}_2$ for the degree--$4$ component on $\widetilde K^0(\Sigma X)$.
A class in $K^{-1}(M)$ may be represented by a smooth map $g:M\to GL(N,\mathbb C)$ for $N$ large enough.

The odd Chern character is the closed odd differential form \cite{ OddChernChar,BottSeely,AlvarezGaume:1985}
\begin{equation}\label{eq:odd_chern_character_def}
    \operatorname{ch}(g)
    =\sum_{n\ge 0}\frac{(-1)^n}{(2\pi i)^{n+1}}\frac{n!}{(2n+1)!}\,
    \Tr\!\big((g^{-1}dg)^{2n+1}\big) \,,
\end{equation}
where $\Tr$ is the \emph{matrix} trace in $GL(N,\mathbb C)$.
Its degree-$3$ component is
\begin{equation}\label{eq:ch3_component}
    \operatorname{ch}_3(g)=\frac{1}{24\pi^2}\Tr\!\big((g^{-1}dg)^3\big) \,.
\end{equation}
When $g:S^3\to U(N)$, one has
\begin{equation*}
    \deg(g)=\int_{S^3}\operatorname{ch}_3(g)\in\mathbb Z \,,
\end{equation*}
and this integer agrees with the class of $g$ in $\pi_3(U)\cong\mathbb Z$.

Higher-dimensional generalizations and index-theoretic interpretations (Dirac/Toeplitz operators)
are not needed for our main arguments; see standard references on Bott periodicity and odd Chern characters.

\subsection{K-theory of Lie groups and the representation ring}

Let $G$ be a simple simply-connected compact Lie group. We denote by $R(G)$ its representation ring. 
The representation ring $R(G)$ is the Grothendieck group of finite-dimensional complex representations of $G$.
Concretely, it is the free abelian group on isomorphism classes $[V]$ of representations, modulo the relations
$[V] + [W] - [V \oplus W] = 0$ 
for all representations $V,W$ of $G$.  (Multiplication is induced by tensor product:
$[V]\cdot[W]=[V\otimes W]$.)

Any unitary representation $\rho:G\to U(\dim\rho)$ defines a class in $K^{-1}(G)$ via its homotopy class: 
\begin{equation*}
\beta(\rho)\in K^{-1}(G)\cong [G,U]
\end{equation*}
This extends to a homomorphism
\begin{equation*}
\beta:\ R(G)\longrightarrow K^{-1}(G).
\end{equation*}
Using the Bott identification~\eqref{eq:Bott_susp_identification}, we define
\begin{equation*}
\Sigma_K:\ K^{-1}(G)\xrightarrow{\ \cong\ }\widetilde K^0(\Sigma G),\qquad
\beta_\Sigma(\rho):=\Sigma_K(\beta(\rho)).
\end{equation*}
Concretely, $\beta_\Sigma(\rho)$ is the reduced bundle over $\Sigma G$ obtained by the clutching construction along
$\rho:G\to U(\dim \rho)$.

The structure of $K(G)$ and $K^0(BG)$ is governed by classical theorems:

\begin{thm}[Hodgkin \cite{Hodgkin}, Atiyah \cite{Atiyah}]\label{thm:Hodgkin-Atiyah}
If $G$ is a simple, simply-connected compact Lie group of rank $\ell$, then $K(G)$ is torsion-free and
\begin{equation*}
K(G)\cong \Lambda\big(\beta(\rho_1),\ldots,\beta(\rho_\ell)\big),
\end{equation*}
an exterior algebra on the primitive classes corresponding to the fundamental representations $\rho_i$.
\end{thm}

\begin{thm}[Atiyah--Segal completion theorem \cite{AtiyahSegal}]\label{thm:Atiyah-Segal}
Let $G$ be compact connected and let $I\subset R(G)$ be the augmentation ideal. Then
\begin{equation*}
    K^{0}(BG)\cong R(G)_{\widehat I} \,.
\end{equation*}
Moreover, for connected $G$ one has $K^{-1}(BG)=0$.
\end{thm}

Altogether, we have the following triangular commutative diagram connecting the representation ring $R(G)$, the odd K-theory group $K^{-1}(G)$ of $G$, and the reduced even K-theory group $\widetilde{K}^0(\Sigma G)$ of the suspension $\Sigma G$:
\begin{equation}
\begin{tikzcd}[row sep=large, column sep=small]
  & R(G) \arrow[dl, "\beta"'] \arrow[dr, "\beta_\Sigma"] &  \\
K^{-1}(G)
  \arrow[rr, "\Sigma_K\;(\cong)"']
  & &
\widetilde{K}^0(\Sigma G) 
\end{tikzcd}\label{eq:BetaCD}
\end{equation}

\subsection{The Dynkin index as a topological characteristic number}

\begin{thm}[Harris \cite{Harris}]\label{thm:Harris}
Let $G$ be a simple, simply-connected compact Lie group, and fix the basic inner product
$\langle\cdot,\cdot\rangle$ on $\mathfrak g$ normalized by $\langle\alpha^\vee,\alpha^\vee\rangle=2$
for every long root $\alpha$. Let $x_3(G)\in H^3(G,\mathbb Z)$ be the positive generator determined by this normalization.
Let $\beta:R(G)\to K^{-1}(G)$ be the $\beta$--construction.
Then for every finite-dimensional complex representation $\rho$ of $G$,
\begin{equation*}
\operatorname{ch}_3(\beta(\rho))=\ell_\rho\,x_3(G)\qquad\text{in }H^3(G,\mathbb Z),
\end{equation*}
where $\ell_\rho$ is the Dynkin index of $\rho$ with respect to $\langle\cdot,\cdot\rangle$.
\end{thm}

Although $\operatorname{ch}$ is a priori rational, for simply-connected simple $G$ the classes
$\operatorname{ch}_3(\beta(\rho))$ land integrally in $H^3(G,\mathbb Z)$ after choosing the generator $x_3$ as above.
Harris' theorem identifies the integral coefficient with the Dynkin index.

We announce the theorem of Naylor \cite{Naylor}  with corrections of a couple of misprints\footnote{
We corrected two misprints in Naylor's theorem:
\begin{enumerate}
\item \emph{Spin groups and low rank isomorphisms.} One should state the Spin case for $Spin(n)$ with $n\ge 7$ to avoid a contradiction due to low-rank isomorphisms. In particular,
$d_{Spin(n)}=2$ for $n\ge 7$, while $d_{Spin(3)}=d_{Spin(5)}=d_{Spin(6)}=1$ since
$Spin(3)\cong SU(2)$, $Spin(5)\cong Sp(2)$ and $Spin(6)\cong SU(4)$.
The original formulation involves $n\geq 5$, which end up giving two distinct values for $Spin(n)$ for $1 \leq n\leq 6$.
\item \emph{$E_8$ value.} The correct value for $d_G$ in the case of $E_8$ is $d_{E_8}=60$ (not $20$) since $60$ is the greatest common divisor of the Dynkin indices of the fundamental representations of $E_8$.
\end{enumerate}
}. 
\begin{thm}[Naylor]\label{thm:Naylor_corrected}
Let $G$ be a simple, simply-connected compact Lie group.
Let $\beta_\Sigma(\rho)\in \widetilde K^0(\Sigma G)$ be the Bott-suspended class associated to a complex
representation $\rho$ of $G$, and let $u_4=\Sigma^\sharp(x_3)$ be the generator of $H^4(\Sigma G,\mathbb Z)$
corresponding to the positive generator $x_3\in H^3(G,\mathbb Z)$ fixed in the main text.
\begin{enumerate}
\item For every finite-dimensional complex representation $\rho$ of $G$,
\begin{equation}\label{eq:Naylor_ch2_corrected}
\operatorname{ch}_2\big(\beta_\Sigma(\rho)\big)\;=\;\ell_\rho\,u_4
\qquad\in H^4(\Sigma G,\mathbb Z),
\end{equation}
where $\ell_\rho$ is the Dynkin index of $\rho$ (normalized as in Section~\ref{sec:DI}).
\item The image of $\operatorname{ch}_2:\widetilde K^0(\Sigma G)\to H^4(\Sigma G,\mathbb Z)\cong \mathbb Z\cdot u_4$
is the subgroup generated by $d_G\,u_4$, where $d_G$ is the gcd of the Dynkin indices of the fundamental
representations. With our normalization one has
\begin{equation*}
d_G=
\begin{cases}
1, & G=SU(n)\ (n\ge 2),\\
1, & G=Sp(n)\ (n\ge 1),\\
2, & G=Spin(n)\ (n\ge 7),\\
2, & G=G_2,\\
6, & G=F_4,\\
6, & G=E_6,\\
12, & G=E_7,\\
60, & G=E_8.
\end{cases}
\end{equation*}
\end{enumerate}
\end{thm}

\medskip

\begin{rem}[Compatibility of Harris and Naylor]\label{rem:Naylor_vs_Harris}
Under the Bott identification $\Sigma_K$ and the suspension isomorphism $\Sigma^\sharp$, the two formulas are the same:
\begin{equation*}
\Sigma^\sharp\!\left(\operatorname{ch}_3(\beta(\rho))\right)
\;=\;
\operatorname{ch}_2\!\left(\beta_\Sigma(\rho)\right)\in H^4(\Sigma G,\mathbb Q),
\end{equation*}
and with the generator choices~\eqref{eq:triangle_commutes} this identity is integral.
\end{rem}

\subsection{The basic triangle in cohomology}

Assume now that $G$ is compact, simple, and simply-connected.
A convenient way to organize the fundamental classes
\begin{equation*}
x_3 \in H^3(G,\mathbb{Z}),\qquad
y_4 \in H^4(BG,\mathbb{Z}),\qquad
u_4 \in H^4(\Sigma G,\mathbb{Z})
\end{equation*}
is via the following triangle commutative diagram:
\begin{equation}\label{eq:BGCD}
\begin{tikzcd}
H^3(G,\mathbb{Z})
  \arrow[rr, "\Sigma^\sharp"]
  & &
H^4(\Sigma G,\mathbb{Z}) \\
&
\arrow[ul, "\tau"]
H^4(BG,\mathbb{Z})
   \arrow[ur, "\sigma^*"'] &
\end{tikzcd}
\end{equation}
Here $\tau$ denotes the transgression in the Serre spectral sequence associated to the universal principal $G$--bundle
$
G \longrightarrow EG \longrightarrow BG$, $\Sigma^\sharp$ is the suspension isomorphism, and $\sigma:\Sigma G\to BG$ is the standard map obtained (up to homotopy)
from the equivalence $G\simeq \Omega BG$ and the counit $\Sigma\Omega BG\to BG$.

Each group in~\eqref{eq:BGCD} is infinite cyclic (Appendix~\ref{sec:A}), and $\tau$ is an isomorphism up to sign.
We choose generators $x_3$, $y_4$, and $u_4$ so that the triangle commutes:
\begin{equation}\label{eq:triangle_commutes}
\tau(y_4)=x_3,\qquad \Sigma^\sharp(x_3)=u_4,\qquad \sigma^*(y_4)=u_4.
\end{equation}
All three classes arise from the same basic invariant quadratic form
\begin{equation*}
Q(X)=\tfrac12(X,X)
\end{equation*}
on $\mathfrak g$, viewed on $BG$, $G$, and $\Sigma G$ respectively.

\medskip

\noindent\textbf{From cohomology back to $\pi_3(G)$.}
For such $G$, the Hurewicz map $h:\pi_3(G)\to H_3(G,\mathbb Z)$ is an isomorphism, and the evaluation pairing
$H^3(G,\mathbb Z)\times H_3(G,\mathbb Z)\to\mathbb Z$ is perfect.
Once $x_3$ is fixed, there is a unique generator $\gamma_3\in\pi_3(G)$ characterized by
$\langle x_3,h(\gamma_3)\rangle=1$.

\subsection{The hexagon: the unifying picture
}\label{sec:hexagon}

The commutative diagrams in equations \eqref{eq:BetaCD} and \eqref{eq:BGCD} are unified by the Chern characters $ch_3$ and $ch_2$, the result is a hexagon illustrated in Figure \ref{fig:hexagon}. 

\begin{figure}[H]
\centering
\begin{equation*}
\begin{tikzcd}[row sep=large, column sep=small]
  & R(G) \arrow[dl, "\beta"'] \arrow[dr, "\beta_\Sigma"] &  \\
K^{-1}(G)
  \arrow[dd, "\operatorname{ch}_3"']
  \arrow[rr, "\Sigma_K\;(\cong)"']
& &
\widetilde{K}^0(\Sigma G) \arrow[dd, "\operatorname{ch}_2"] \\
& &  \\
H^3(G,\mathbb{Z})
  \arrow[rr, "\Sigma^\sharp"']
& &
H^4(\Sigma G,\mathbb{Z}) \\
&  \arrow[ul, "\tau"'] H^4(BG,\mathbb{Z}) \arrow[ur, "\sigma^*"'] &
\end{tikzcd}
\end{equation*}
\caption{The $K$--theory/cohomology hexagon for a simple simply-connected compact Lie group $G$.
All arrows are natural in $G$; Harris and Naylor identify the Dynkin index as the coefficient of the primitive generators.
Naturality is the mechanism behind the universal scaling theorem for embeddings.}
\label{fig:hexagon}
\end{figure}

The diagram on Figure \ref{fig:hexagon} is a functorial compatibility statement:
\begin{itemize}
\item the two maps $\beta$ and $\beta_\Sigma$ are identified by Bott periodicity ($\Sigma_K$);
\item the bottom triangle encodes the three cohomology incarnations of the same quadratic form $Q$;
\item Harris and Naylor identify the Dynkin index $\ell_\rho$ as the coefficient of the primitive generators:
\begin{equation*}
\operatorname{ch}_3(\beta(\rho))=\ell_\rho\,x_3,\qquad
\operatorname{ch}_2(\beta_\Sigma(\rho))=\ell_\rho\,u_4.
\end{equation*}
\end{itemize}

In the next section we use naturality of every arrow in the hexagon to show that for an embedding $f:G\hookrightarrow H$
a \emph{single} integer $j_f$ rescales simultaneously $y_4$, $x_3$, $u_4$, and $\pi_3$, and that this integer is exactly
the Dynkin embedding index.

\section{The universal scaling theorem}\label{sec:universalscaling}

In this section we explain the \emph{topological} meaning of the Dynkin embedding index.
For an embedding of simple, simply-connected compact Lie groups
\begin{equation*}
f:\;G\hookrightarrow H,
\end{equation*}
we show that there is a single integer $j_f$ governing the induced maps on
\begin{equation*}
H^4(B(-),\mathbb Z),\qquad H^3((-),\mathbb Z),\qquad H^4(\Sigma(-),\mathbb Z),\qquad \pi_3((-)),
\end{equation*}
and we relate this same integer to restriction of representations via the Chern character.

\subsection{Low-degree topology of G and BG}

Let $G$ be simple and simply-connected.
Bott periodicity gives $\pi_2(G)=0$ and $\pi_3(G)\cong \mathbb Z$, hence $G$ is $2$-connected.
\noindent
Since $G\simeq \Omega BG$, we have the standard shift of homotopy groups
\begin{equation}
\pi_i(BG)\;\cong\;\pi_{i-1}(G)\qquad (i\ge 2).
\end{equation}
For $G$ simple and simply-connected, Cartan's theorem gives $\pi_1(G)=\pi_2(G)=0$, hence
\begin{equation}
\pi_2(BG)=\pi_1(G)=0,\qquad \pi_3(BG)=\pi_2(G)=0,\qquad \pi_1(BG)=0.
\end{equation}
Therefore $BG$ is $3$-connected.

In particular (Appendix~\ref{sec:A}),
\begin{equation}\label{eq:lowdeg_G}
H^1(G,\mathbb Z)=H^2(G,\mathbb Z)=0,\qquad H^3(G,\mathbb Z)\cong \mathbb Z,
\end{equation}
and
\begin{equation}\label{eq:lowdeg_BG}
H^1(BG,\mathbb Z)=H^2(BG,\mathbb Z)=H^3(BG,\mathbb Z)=0,\qquad H^4(BG,\mathbb Z)\cong \mathbb Z.
\end{equation}
We choose generators $x_3(G)\in H^3(G,\mathbb Z)$ and $y_4(G)\in H^4(BG,\mathbb Z)$ compatibly with transgression,
\begin{equation}\label{eq:transgression_x3_y4}
\tau\big(y_4(G)\big)=x_3(G),
\end{equation}
and set
\begin{equation}\label{eq:def_u_generator}
u_4(G):=\Sigma^\sharp\big(x_3(G)\big)\in H^4(\Sigma G,\mathbb Z).
\end{equation}

\subsection{The embedding index as a universal scaling integer}

Let $f:G\hookrightarrow H$ be an embedding of simple, simply-connected compact Lie groups.
It induces maps
\begin{align}
& f^*: H^3(H, \mathbb{Z})\to H^3(G,\mathbb{Z}),\label{eq:induced_f_H3}\\
& (Bf)^*: H^4(BH, \mathbb{Z})\to H^4(BG,\mathbb{Z}),\label{eq:induced_Bf_H4}\\
& (\Sigma f)^*: H^4(\Sigma H,\mathbb{Z})\to H^4(\Sigma G,\mathbb{Z}),\label{eq:induced_Sigmaf_H4}
\end{align}
and on homotopy
\begin{equation}\label{eq:induced_map_pi3}
f_*:\pi_3(G)\to \pi_3(H).
\end{equation}
Each source and target above is infinite cyclic, so each induced map is multiplication by an integer once generators are chosen.

\begin{thm}[Universal scaling]\label{thm:universal_scaling}
Let $f\colon G\hookrightarrow H$ be an embedding of compact, simple, simply-connected Lie groups.
Let $(\cdot,\cdot)_G$ and $(\cdot,\cdot)_H$ be the \emph{reduced} invariant inner products
(i.e.\ $\Ad$-invariant, positive definite on the compact real forms, and normalized so that every long root has squared length $2$).
There is a unique integer $j_f\ge 1$ such that
\begin{equation}\label{eq:def_jf_inner_product}
\Big( df(X),df(Y)\Big)_H \;=\; j_f\, (X,Y)_G
\qquad (X,Y\in\mathfrak g).
\end{equation}
(Equivalently, $j_f$ is the Dynkin embedding index.)

For each such $K\in\{G,H\}$, fix:
\begin{itemize}
\item the canonical generator $\gamma_3(K)\in\pi_3(K)\cong\mathbb Z$ given by the highest-root embedding
$SU(2)\hookrightarrow K$ (Appendix~B),
\item the dual generator $x_3(K)\in H^3(K;\mathbb Z)\cong\mathbb Z$ characterized by
$\Big\langle x_3(K),h(\gamma_3(K))\Big\rangle=1$ (with $h$ the Hurewicz isomorphism),
\item $u_4(K):=\Sigma^\sharp(x_3(K))\in H^4(\Sigma K;\mathbb Z)$ (reduced suspension),
\item the generator $y_4(K)\in H^4(BK;\mathbb Z)$ characterized by $\tau\!\big(y_4(K)\big)=x_3(K)$.
\end{itemize}
Then the maps induced by $f$ satisfy
\begin{align*}
f_*\big(\gamma_3(G)\big) &= j_f\,\gamma_3(H) \qquad\text{in }\pi_3(H),\\
f^*\big(x_3(H)\big) &= j_f\,x_3(G) \qquad\text{in }H^3(G;\mathbb Z),\\
(\Sigma f)^*\big(u_4(H)\big) &= j_f\,u_4(G) \qquad\text{in }H^4(\Sigma G;\mathbb Z),\\
(Bf)^*\big(y_4(H)\big) &= j_f\,y_4(G) \qquad\text{in }H^4(BG;\mathbb Z).
\end{align*}
In particular, each of these homomorphisms is injective with cokernel $\mathbb Z/j_f\mathbb Z$. Moreover, for any complex representation $\rho$ of $H$ such that
$\operatorname{ch}_3(\beta(\rho))=\ell_\rho\,x_3(H)$, one has
$\ell_{\rho\circ f}=j_f\,\ell_\rho$ where $\beta(\rho)\in K^{-1}(H)$  is the class associated to the map $H\to U$ defined by $\rho$.
\end{thm}

\begin{proof}
Assume $G$ and $H$ are compact, connected, simply connected, simple Lie groups.
Let $\eta_K$ be the bi-invariant Cartan $3$-form on $K$ associated to the reduced inner product
(Appendix~B), normalized so that $[\eta_K]=x_3(K)\in H^3(K;\mathbb Z)$ and
\begin{equation}
\int_{SU(2)_\theta}\eta_K=1,
\end{equation}
where $SU(2)_\theta\subset K$ is the subgroup corresponding to the highest root.
Let $\gamma_3(K)\in \pi_3(K)\cong \mathbb Z$ be the generator represented by the inclusion
$SU(2)_\theta\hookrightarrow K$, and let $h:\pi_3(K)\to H_3(K;\mathbb Z)$ be the Hurewicz map.

\smallskip
\emph{Step 1: homotopy via the Cartan $3$-form.}
By \eqref{eq:def_jf_inner_product} and functoriality of the Cartan construction,
\begin{equation}
f^*\eta_H = j_f\,\eta_G
\end{equation}
as differential forms on $G$. Evaluating on the fundamental class $h(\gamma_3(G))$ gives
\begin{equation}
\langle f^*x_3(H),\,h(\gamma_3(G))\rangle
=\int_{h(\gamma_3(G))} f^*\eta_H
=j_f\int_{h(\gamma_3(G))}\eta_G
=j_f.
\end{equation}
By naturality of the Kronecker pairing,
$
\langle f^*x_3(H),h(\gamma_3(G))\rangle=\langle x_3(H),h(f_*\gamma_3(G))\rangle$.
Since $h(\gamma_3(H))$ is characterized by $\langle x_3(H),h(\gamma_3(H))\rangle=1$,
this forces $f_*(\gamma_3(G))=j_f\,\gamma_3(H)$.

\smallskip
\emph{Step 2: cohomology in degree $3$.}
Equivalently, $f^*(x_3(H))=j_f\,x_3(G)$ in $H^3(G;\mathbb Z)$.

\smallskip
\emph{Step 3: suspension.}
Naturality of the reduced cohomology suspension $\Sigma^\sharp$ gives
\begin{equation}
(\Sigma f)^*(u_4(H))
=(\Sigma f)^*\Sigma^\sharp(x_3(H))
=\Sigma^\sharp f^*(x_3(H))
=\Sigma^\sharp(j_f x_3(G))
=j_f\,u_4(G).
\end{equation}

\smallskip
\emph{Step 4: transgression.}
In the universal bundles, naturality of transgression yields
\begin{equation}
f^*(x_3(H))
=f^*(\tau(y_4(H)))
=\tau\bigl((Bf)^*(y_4(H))\bigr).
\end{equation}
With our choice of generators satisfying $\tau(y_4(G))=x_3(G)$ and $\tau(y_4(H))=x_3(H)$,
the identity $f^*(x_3(H))=j_f\,x_3(G)$ forces
\begin{equation}
(Bf)^*(y_4(H))=j_f\,y_4(G).
\end{equation}

\smallskip
\emph{Finally, the representation consequence.}
Naturality of $\beta$ and of the Chern character gives
\begin{equation}
\operatorname{ch}_3\!\bigl(\beta(\rho\circ f)\bigr)
=f^*\operatorname{ch}_3\!\bigl(\beta(\rho)\bigr)
=f^*(\ell_\rho x_3(H))
=\ell_\rho\,j_f\,x_3(G),
\end{equation}
hence $\ell_{\rho\circ f}=j_f\,\ell_\rho$.
\end{proof}

\medskip

Embeddings with $j_f=1$ are therefore precisely those for which the induced maps
in degrees $3$ and $4$ are isomorphisms; equivalently, no ``fractionalization'' occurs in these basic topological invariants.

\subsection{Third homotopy group of the  homogeneous space H/G}

Since $G$ is compact, $f(G)\subset H$ is closed; hence $H/G$ is a smooth homogeneous space and the projection
$H\to H/G$ is a principal $G$-bundle with basis $H/G$.

\begin{prop}\label{prop:pi3_H_over_G}
Let $f:G\hookrightarrow H$ be an embedding of simple, simply-connected compact Lie groups with Dynkin embedding index $j_f$.
Then
\begin{equation*}
\pi_3(H/G)\;\cong\;\operatorname{coker}\big(f_*:\pi_3(G)\to \pi_3(H)\big)\;\cong\;\mathbb Z/j_f\mathbb Z.
\end{equation*}
\end{prop}

\begin{proof}
Consider the principal $G$-bundle $G\to H\to H/G$.
The long exact homotopy sequence contains
\begin{equation*}
\pi_3(G)\xrightarrow{\,f_*\,}\pi_3(H)\longrightarrow \pi_3(H/G)\longrightarrow \pi_2(G).
\end{equation*}
Since $\pi_2(G)=0$, we have $\pi_3(H/G)\cong \operatorname{coker}(f_*)$.
Finally $\pi_3(G)\cong \pi_3(H)\cong \mathbb Z$, and by Theorem~\ref{thm:universal_scaling} the map $f_*$ is multiplication by $j_f$ on generators,
so the cokernel is $\mathbb Z/j_f\mathbb Z$.
\end{proof}

\subsection{Naturality in K-theory: scaling laws for Dynkin indices}\label{subsec:K_naturality_scaling}

We now make explicit the functorial content of the hexagon~ given in Figure~\ref{fig:hexagon}.
All arrows there are natural with respect to maps of spaces, and in particular with respect to a group embedding $f:G\hookrightarrow H$.

\medskip

\noindent\textbf{Naturality of $\beta$ and of the Chern character.}
Let $\rho:H\to U(N)$ be a unitary representation. Its class
$\beta_H(\rho)\in K^{-1}(H)\cong [H,U]$ pulls back to the restricted representation:
\begin{equation*}
f^!\big(\beta_H(\rho)\big)=\beta_G(\rho\circ f)\in K^{-1}(G).
\end{equation*}
Moreover, the Chern character is natural:
for any $g:X\to Y$ and any $a\in K^*(Y)$,
\begin{equation*}
\operatorname{ch}(g^!a)=g^*\operatorname{ch}(a)\in H^*(X;\mathbb Q).
\end{equation*}
Applying this with $g=f$ (and with $g=\Sigma f$) yields
\begin{align}
\operatorname{ch}_3\!\left(\beta_G(\rho\circ f)\right)
&=f^*\operatorname{ch}_3\!\left(\beta_H(\rho)\right),\label{eq:ch3_natural}\\
\operatorname{ch}_2\!\left(\beta_{\Sigma,G}(\rho\circ f)\right)
&=(\Sigma f)^*\operatorname{ch}_2\!\left(\beta_{\Sigma,H}(\rho)\right).\label{eq:ch2_natural}
\end{align}

\medskip

\noindent\textbf{Scaling of Dynkin indices.}
By Harris (and equivalently Naylor),
\begin{equation*}
\operatorname{ch}_3\!\left(\beta_H(\rho)\right)=\ell_\rho\,x_3(H),\qquad
\operatorname{ch}_3\!\left(\beta_G(\rho\circ f)\right)=\ell_{\rho\circ f}\,x_3(G).
\end{equation*}
Using $f^*x_3(H)=j_f\,x_3(G)$ from Theorem~\ref{thm:universal_scaling} and \eqref{eq:ch3_natural} gives
\begin{equation}\label{eq:ell_scaling_formula}
\ell_{\rho\circ f}=j_f\,\ell_\rho.
\end{equation}
Equivalently, applying the same argument to \eqref{eq:ch2_natural} and $(\Sigma f)^*u_4(H)=j_f\,u_4(G)$ yields the identical scaling law via $\operatorname{ch}_2$.

\medskip

\noindent\textbf{Index-one criterion.}
In particular,
\begin{equation*}
j_f=1\quad\Longleftrightarrow\quad
\ell_{\rho\circ f}=\ell_\rho\ \text{ for all representations }\rho\text{ of }H,
\end{equation*}
and equivalently $f^*:H^3(H;\mathbb Z)\xrightarrow{\ \cong\ }H^3(G;\mathbb Z)$ (and likewise in degree $4$ on $BG$ and $\Sigma G$).

\subsection{The Dynkin index as a characteristic number}

To emphasize the ``characteristic number'' interpretation, it is convenient to normalize $x_3(G)$ using an index-one $SU(2)\subset G$.
For simply-connected simple $G$, there is (up to conjugacy) a distinguished embedding
$\varphi:SU(2)\hookrightarrow G$ of Dynkin embedding index $1$ corresponding to the highest root; see
\cite[\S 8]{AHS}, \cite{Bernard}, and \cite[Appendix A]{Csaki}.
Topologically $SU(2)\cong S^3$, and $\varphi$ represents the generator of $\pi_3(G)$.

We fix the sign of $x_3(G)$ by requiring
\begin{equation}\label{eq:normalize_x3_by_SU2}
\int_{S^3}\varphi^*x_3(G)=1.
\end{equation}

\begin{thm}\label{thm:ell_as_characteristic_number}
Let $G$ be simply-connected and simple, and let $\rho$ be a finite-dimensional complex representation of $G$.
Let $S^3_*\subset G$ denote the image of an index-one subgroup $\varphi:SU(2)\hookrightarrow G$.
Then
\begin{equation*}
\ell_\rho \;=\;\int_{S^3_*}\operatorname{ch}_3\big(\beta(\rho)\big).
\end{equation*}
Equivalently,
\begin{equation*}
\ell_\rho=\int_{\Sigma S^3_*}\operatorname{ch}_2\big(\beta_\Sigma(\rho)\big).
\end{equation*}
\end{thm}

\begin{proof}
By Harris, $\operatorname{ch}_3(\beta(\rho))=\ell_\rho\,x_3(G)$.
Integrating over $S^3_*$ and using~\eqref{eq:normalize_x3_by_SU2} gives the first identity.
The second follows from $\Sigma^\sharp(\operatorname{ch}_3(\beta(\rho)))=\operatorname{ch}_2(\beta_\Sigma(\rho))$ and $u_4(G)=\Sigma^\sharp(x_3(G))$.
\end{proof}

\subsection{Dynkin embedding index and decay of minimally charged instantons}

For a Yang--Mills theory with gauge algebra $\mathfrak g$, the kinetic and topological terms involve
an invariant bilinear form $(\cdot,\cdot)_{\mathfrak g}$ and a coupling constant $g_{\mathfrak g}$:
\begin{equation*}
\frac{1}{2g^2_{\mathfrak g}}\int (F,*F)_{\mathfrak g},\qquad
\frac{1}{2g^2_{\mathfrak g}}\int (F,F)_{\mathfrak g}.
\end{equation*}
If $f:\mathfrak g\hookrightarrow\mathfrak h$ has Dynkin embedding index $j_f$, then by definition
\begin{equation*}
(fX,fY)_{\mathfrak h}=j_f\,(X,Y)_{\mathfrak g}.
\end{equation*}
Thus the terms above rescale uniformly, and one may reinterpret this as a rescaling of the coupling:
\begin{equation*}
\frac{1}{g^2_{\mathfrak g}}=\frac{j_f}{g^2_{\mathfrak h}}.
\end{equation*}
In particular, a unit-charge instanton in the $\mathfrak g$-theory corresponds to $j_f$ units of charge in the parent $\mathfrak h$-theory.
When $j_f>1$ this is often described as the existence of $\mathbb Z_{j_f}$-instantons; see \cite{Csaki} following \cite{Intriligator:1995id}.

\subsection{The index-one selection rule as a conservation law}

By Theorem~\ref{thm:universal_scaling}, the condition $j_f=1$ is equivalent to requiring that the induced maps
\begin{align*}
(Bf)^* &: H^4(BH,\mathbb Z)\longrightarrow H^4(BG,\mathbb Z),\\
(\Sigma f)^* &: H^4(\Sigma H,\mathbb Z)\longrightarrow H^4(\Sigma G,\mathbb Z),\\
f^* &: H^3(H,\mathbb Z)\longrightarrow H^3(G,\mathbb Z),\\
f_* &: \pi_3(G)\longrightarrow \pi_3(H),
\end{align*}
are all isomorphisms.  In that case there is no ``loss'' or ``fractionalization'' of instanton charge, Chern--Simons/WZW levels,
or the suspended class detected by $\operatorname{ch}_2$.

To a given embedding $f:\mathfrak g\hookrightarrow\mathfrak h$ one can associate an embedding matrix $U_f$ as in Lorente--Gruber \cite{LorenteGruber};
it satisfies
\begin{equation*}
U_fU_f^\dagger=j_f\,\mathrm{Id},
\end{equation*}
so the rows of $U_f$ are orthonormal (with respect to the standard Hermitian form) if and only if $j_f=1$.

Combining these observations with~\eqref{eq:ell_scaling_formula} yields equivalent characterizations of index-one embeddings:
\begin{itemize}
\item[(1)] \textbf{Instanton spectrum is preserved.} A minimal instanton in $G$ remains minimal in $H$ iff $j_f=1$.
\item[(2)] \textbf{Chern--Simons levels are preserved.} The classification by $H^4(BG,\mathbb Z)$ is unchanged iff $j_f=1$.
\item[(3)] \textbf{WZW levels are preserved.} The classification by $H^3(G,\mathbb Z)$ is unchanged iff $j_f=1$.
\item[(4)] \textbf{Suspended class is preserved.} The classification by $H^4(\Sigma G,\mathbb Z)$ (equivalently the $\operatorname{ch}_2$ image) is unchanged iff $j_f=1$.
\item[(5)] \textbf{Embedding matrix is orthonormal.} The Lorente--Gruber matrix $U_f$ has orthonormal rows iff $j_f=1$.
\item[(6)] \textbf{Dynkin indices are preserved under restriction.} $\ell_{\rho\circ f}=\ell_\rho$ for all $\rho$ iff $j_f=1$.
\end{itemize}
Together these show that the ``index-one selection rule'' is precisely a conservation law for the basic topological charges and levels under embeddings of gauge groups.

\subsection*{Acknowledgements}

The authors are grateful to Saber Ahmed, Sungkyung Kang, and Alfred No\"el for discussions. Part of this work was carried out while M.E. was a visitor at the Simons Laufer Mathematical Sciences Institute (SLMath) in Berkeley, California.
M.~J.~K.~is supported by the Start-up Research Grant for new faculty provided by Texas A\&M University.

\begin{appendices}

\section{Low-degree topology of a compact Lie group and its classifying space}\label{sec:A}

This appendix collects standard facts about invariant bilinear forms, Lie algebra cohomology,
and the low-degree homotopy and (co)homology of compact, simple, simply-connected Lie groups
and their classifying spaces.  We also recall the suspension--loop adjunction and the basic
compatibility between transgression and suspension that we use later.

\subsection{Invariant bilinear forms and Lie algebra cohomology}

Let $\mathfrak g$ be a (real or complex) Lie algebra.
An \emph{$\Ad$-invariant symmetric bilinear form} on $\mathfrak g$ is a symmetric bilinear form
$B:\mathfrak g\times \mathfrak g\to \mathbb C$ such that for all $x,y,z\in\mathfrak g$,
\begin{equation}\label{eq:inv_form_appendix}
B([z,x],y)+B(x,[z,y])=0.
\end{equation}

\begin{thm}[See e.g.\ {\cite[Prop.~2.3.2]{Loop}}]\label{thm:inv_forms_one_dim}
Let $\mathfrak g$ be a complex simple Lie algebra. Then the space of invariant symmetric bilinear forms
on $\mathfrak g$ is one-dimensional. In particular, any two invariant symmetric bilinear forms differ
by an overall scalar factor.
\end{thm}

\medskip
\noindent\textbf{The Killing form.}
The Killing form $\kappa(x,y):=\mathrm{tr}(\mathrm{ad}(x)\mathrm{ad}(y))$ is an invariant symmetric bilinear form
on any Lie algebra, and it is nondegenerate when $\mathfrak g$ is semisimple.

\medskip
\noindent\textbf{Lie algebra cohomology and de Rham cohomology.}
Let $G$ be a compact, connected Lie group with Lie algebra $\mathfrak g$.
Write $\Omega^*(G)^{\mathrm L}$ for the complex of left-invariant differential forms on $G$.
Evaluation at the identity identifies $\Omega^*(G)^{\mathrm L}$ with the Chevalley--Eilenberg complex
$\Lambda^*\mathfrak g^*$, and the de Rham differential corresponds to the Chevalley--Eilenberg differential.

\begin{thm}[Cartan--Chevalley--Eilenberg]\label{thm:CE_equals_dR_compact}
If $G$ is compact and connected, then the inclusion
$\Omega^*(G)^{\mathrm L}\hookrightarrow \Omega^*(G)$ induces an isomorphism
\begin{equation*}
H^*_{\mathrm{CE}}(\mathfrak g;\mathbb R)\;\cong\;H^*_{\mathrm{dR}}(G;\mathbb R).
\end{equation*}
\end{thm}

\begin{rem}
A convenient proof of Theorem~\ref{thm:CE_equals_dR_compact} uses a bi-invariant Riemannian metric on $G$:
every de Rham cohomology class has a unique harmonic representative, and harmonic forms are automatically
bi-invariant, hence left-invariant.  Thus every class is represented by a left-invariant form.
\end{rem}

\medskip
\noindent\textbf{Koszul's cocycle.}
Given an invariant symmetric bilinear form $B$ on $\mathfrak g$, define the $3$-cochain
\begin{equation}\label{eq:koszul_cocycle}
\mu_B(x,y,z):=B(x,[y,z]).
\end{equation}
Invariance \eqref{eq:inv_form_appendix} implies that $\mu_B$ is a Chevalley--Eilenberg $3$-cocycle.

\begin{thm}[Koszul {\cite[\S 11]{Koszul}}]\label{thm:Koszul_map}
Let $\mathfrak g$ be a Lie algebra over a field of characteristic zero.
\begin{enumerate}
\item The assignment $B\mapsto [\mu_B]$ defines a linear map
\begin{equation*}
(S^2\mathfrak g^*)^{\mathfrak g}\longrightarrow H^3_{\mathrm{CE}}(\mathfrak g).
\end{equation*}
\item If $H^1_{\mathrm{CE}}(\mathfrak g)=0$, then $B\mapsto [\mu_B]$ is injective.
\item If $\mathfrak g$ is (complex) semisimple, then $H^1_{\mathrm{CE}}(\mathfrak g)=H^2_{\mathrm{CE}}(\mathfrak g)=0$
(Whitehead lemmas), hence $B\mapsto [\mu_B]$ is injective.
\item If $\mathfrak g$ is (complex) simple, then $(S^2\mathfrak g^*)^{\mathfrak g}$ is one-dimensional
(Theorem~\ref{thm:inv_forms_one_dim}) and $H^3_{\mathrm{CE}}(\mathfrak g)$ is also one-dimensional; in particular,
the Koszul map is an isomorphism in this case.
\end{enumerate}
\end{thm}

\subsection{Low-degree homotopy and homology of G}

\begin{thm}[Cartan {\cite{Cartan1936}}]\label{thm:Cartan_pi1_pi2}
If $G$ is compact, connected, and simply-connected, then
\begin{equation*}
\pi_1(G)=0,\qquad \pi_2(G)=0.
\end{equation*}
\end{thm}

\begin{thm}[Bott {\cite{Bott}}]\label{thm:Bott_pi3}
If $G$ is compact, simple, and simply-connected, then
\begin{equation*}
\pi_3(G)\cong \mathbb Z.
\end{equation*}
\end{thm}

\begin{thm}[Hurewicz (connectivity form)]\label{thm:Hurewicz_connectivity}
If $X$ is $(n-1)$-connected with $n\ge 2$, then $H_i(X,\mathbb Z)=0$ for $i<n$ and
the Hurewicz map $\pi_n(X)\to H_n(X,\mathbb Z)$ is an isomorphism.
\end{thm}

Applying Theorem~\ref{thm:Hurewicz_connectivity} to $X=G$ (which is $2$-connected by
Theorem~\ref{thm:Cartan_pi1_pi2}) yields:

\begin{thm}\label{thm:H3_of_G_appendix}
If $G$ is compact, simple, and simply-connected, then
\begin{equation*}
H_1(G,\mathbb Z)=0,\qquad H_2(G,\mathbb Z)=0,\qquad H_3(G,\mathbb Z)\cong \mathbb Z,
\end{equation*}
and the Hurewicz map $\pi_3(G)\to H_3(G,\mathbb Z)$ is an isomorphism.
\end{thm}

\begin{rem}[Real versus integral cohomology]\label{rem:real_vs_integral}
If two compact connected Lie groups have the same Lie algebra, then their real cohomology rings agree,
but their integral cohomology may differ by torsion. For example
$SU(2)\cong S^3$ and $SO(3)\cong \mathbb RP^3$ satisfy
\begin{equation*}
H^2(SU(2),\mathbb Z)=0,\qquad H^2(SO(3),\mathbb Z)\cong \mathbb Z/2\mathbb Z.
\end{equation*}
\end{rem}

\subsection{Classifying space BG and its low-degree (co)homology}\label{app:BG}

The \emph{classifying space} $BG$ is (by definition) the base of a contractible principal $G$--bundle
\begin{equation}
\begin{tikzcd}
G \arrow[r, hook] & EG \arrow[d] \\
& BG .
\end{tikzcd}
\end{equation}
For a CW--complex $X$, principal $G$--bundles on $X$ are classified by based homotopy classes of maps:
\begin{equation}
\mathrm{Prin}_G(X)\;\cong\;[X,BG],
\end{equation}
via pullback of the universal bundle.

Since $EG$ is contractible, the long exact sequence of homotopy groups gives
\begin{equation}\label{eq:pi_shift_BG}
\pi_{n+1}(BG)\cong \pi_{n}(G),\qquad (n\ge 1).
\end{equation}

\begin{thm}\label{thm:low_homotopy_BG}
Let $G$ be compact, simple, and simply-connected. Then
\begin{equation*}
\pi_2(BG)=\pi_3(BG)=0,\qquad \pi_4(BG)\cong \mathbb Z.
\end{equation*}
\end{thm}

\begin{proof}
Using \eqref{eq:pi_shift_BG} and Theorems~\ref{thm:Cartan_pi1_pi2} and \ref{thm:Bott_pi3},
we have $\pi_2(BG)\cong \pi_1(G)=0$, $\pi_3(BG)\cong \pi_2(G)=0$, and $\pi_4(BG)\cong \pi_3(G)\cong \mathbb Z$.
\end{proof}

\begin{thm}\label{thm:low_cohom_BG}
Let $G$ be compact, simple, and simply-connected. Then
\begin{equation*}
H^1(BG,\mathbb Z)=H^2(BG,\mathbb Z)=H^3(BG,\mathbb Z)=0,\qquad H^4(BG,\mathbb Z)\cong \mathbb Z.
\end{equation*}
\end{thm}

\begin{proof}
By Theorem~\ref{thm:low_homotopy_BG}, the space $BG$ is $3$-connected and $\pi_4(BG)\cong \mathbb Z$.

Hurewicz implies $H_i(BG,\mathbb Z)=0$ 
for $i\le 3$ and $H_4(BG,\mathbb Z)\cong \pi_4(BG)\cong \mathbb Z$.
Universal coefficients then give the stated cohomology groups.
\end{proof}

\subsection{Transgression and cohomology suspension}

Let $G$ be a compact, simple, simply-connected Lie group, and let
\begin{equation*}
G \longrightarrow EG \xrightarrow{\ \pi\ } BG
\end{equation*}
be the universal principal $G$--bundle. Fix a basepoint $b_0\in BG$ and write
$F:=\pi^{-1}(b_0)\cong G$ for the fiber over $b_0$.

\subsection*{A definition using relative cohomology}
The map $\pi:(EG,F)\to (BG,b_0)$ induces
\begin{equation*}
\pi^*:\ \widetilde H^n(BG;\mathbb Z)=H^n(BG,b_0;\mathbb Z)\longrightarrow H^n(EG,F;\mathbb Z).
\end{equation*}
Since $EG$ is contractible, the long exact sequence of the pair $(EG,F)$ gives an isomorphism
\begin{equation*}
\delta:\ H^{n-1}(F;\mathbb Z)\xrightarrow{\ \cong\ } H^n(EG,F;\mathbb Z).
\end{equation*}
We define the \emph{inverse transgression} (sometimes also called transgression, depending on convention) by
\begin{equation}\label{eq:invtrans_def}
\tau:\ \widetilde H^n(BG;\mathbb Z)\longrightarrow H^{n-1}(G;\mathbb Z),
\qquad
\tau:=\delta^{-1}\circ \pi^*,
\end{equation}
where we identify $F\cong G$.

In particular, for $n=4$ this yields a canonical isomorphism
$\widetilde H^4(BG;\mathbb Z)\cong H^3(G;\mathbb Z)$, and (after fixing signs) we will arrange that
\begin{equation*}
\tau\big(y_4(G)\big)=x_3(G).
\end{equation*}

\subsection*{Serre spectral sequence viewpoint}
Consider the Serre spectral sequence for the fibration $G\to EG\to BG$ with integral coefficients.
Since $BG$ is simply-connected, the local coefficient system is trivial and
\begin{equation*}
E_2^{p,q}\cong H^p(BG;\mathbb Z)\otimes H^q(G;\mathbb Z).
\end{equation*}
Because $EG$ is contractible, $E_\infty^{p,q}=0$ for $p+q>0$. In low degrees this forces an isomorphism
\begin{equation*}
d_4:\ E_4^{0,3}\cong H^3(G;\mathbb Z)\xrightarrow{\ \cong\ }E_4^{4,0}\cong H^4(BG;\mathbb Z),
\end{equation*}
unique up to an overall sign convention. With the standard sign choice, $d_4(x_3(G))=y_4(G)$, and
the map \eqref{eq:invtrans_def} is the inverse of this differential.

\subsection*{Compatibility with cohomology suspension}
Choose a based homotopy equivalence $\phi:G\simeq \Omega BG$.
Let $\sigma:\Sigma G\to BG$ be the adjoint of $\phi$ under the suspension--loop adjunction.
Write
\begin{equation*}
\Sigma^\sharp:\ H^3(G;\mathbb Z)\xrightarrow{\ \cong\ } \widetilde H^4(\Sigma G;\mathbb Z)
\end{equation*}
for the (reduced) cohomology suspension isomorphism. Then the fundamental compatibility is
\begin{equation}\label{eq:trans_susp_compat}
\sigma^*(y_4(G))=\Sigma^\sharp(x_3(G))\ \in\ \widetilde H^4(\Sigma G;\mathbb Z),
\end{equation}
which is another way to encode the identification between $y_4(G)$ and $x_3(G)$.

\subsection{Reduced suspension, loop space, and transgression}\label{app:SuspLoopAdj}

All spaces in this subsection are \emph{based} (pointed), and all maps preserve base points.
We write $\Sigma X$ for the \emph{reduced suspension} of a based space $X$, and $\Omega Y$ for the based loop space.

There is a natural bijection on based homotopy classes (the \emph{suspension--loop adjunction})
\begin{equation}\label{eq:SigmaOmegaAdjunction}
[\Sigma X,\,Y]_* \;\cong\; [X,\,\Omega Y]_*,
\end{equation}
natural in both $X$ and $Y$.

Let $G$ be a well-pointed topological group and $BG$ its classifying space.
A basic fact is that $G$ is (based) homotopy equivalent to the loop space of its classifying space:
\begin{equation}\label{eq:G_is_loop_BG}
G \;\simeq\; \Omega BG.
\end{equation}
Choose once and for all a based homotopy equivalence $\phi:G\to\Omega BG$.
Then the \textbf{standard map} $\sigma:\Sigma G\to BG$ is defined (up to based homotopy) by
\begin{equation}\label{eq:sigmaDef}
\sigma \;:=\; \varepsilon_{BG}\circ \Sigma(\phi)\;:\;\Sigma G \longrightarrow \Sigma\Omega BG \longrightarrow BG,
\end{equation}
where $\varepsilon_{BG}:\Sigma\Omega BG\to BG$ is the counit of the adjunction.
Equivalently, $\sigma$ is the adjoint (under \eqref{eq:SigmaOmegaAdjunction}) of $\phi$.

\medskip
\noindent\textbf{Transgression versus suspension.}
For the universal principal $G$--bundle $G\to EG\to BG$, the Serre spectral sequence defines a transgression
\begin{equation*}
\tau:\;H^{n+1}(BG;\mathbb Z)\longrightarrow H^{n}(G;\mathbb Z).
\end{equation*}
On the other hand, reduced suspension gives an isomorphism
\begin{equation*}
\Sigma^\sharp:\;H^{n}(G;\mathbb Z)\longrightarrow H^{n+1}(\Sigma G;\mathbb Z).
\end{equation*}
A standard compatibility statement is that, with the usual sign conventions,
\begin{equation}\label{eq:TransgressionVsSuspension}
\Sigma^\sharp\circ \tau \;=\; \sigma^*\;:\;H^{n+1}(BG;\mathbb Z)\longrightarrow H^{n+1}(\Sigma G;\mathbb Z),
\end{equation}
where $\sigma$ is defined by \eqref{eq:sigmaDef}. Equivalently, one has the commutative triangle
\begin{equation}\label{eq:Transgression_Suspension}
\begin{tikzcd}
H^{n}(G;\mathbb Z)
  \arrow[rr,"\Sigma^\sharp"] & &
H^{n+1}(\Sigma G;\mathbb Z) \\
&
H^{n+1}(BG;\mathbb Z)
 \arrow[ul,"\tau"]
 \arrow[ur,"\sigma^*"'] &
\end{tikzcd}
\end{equation}
We will mainly use the case $n=3$.

\section{Consistent normalizations and integral generators}
\label{sec:AppB}

In the topology of Lie groups, it is essential to have consistent normalizations between different classes of objects that are related by a sophisticated network of maps that often form commutative diagrams.  
For a compact, simple, simply-connected Lie group $G$, the groups
$\pi_3(G)$, $H_3(G,\mathbb Z)$, $H^3(G,\mathbb Z)$, and $H^4(BG,\mathbb Z)$
are all infinite cyclic and are linked by standard functorial maps: the Hurewicz
isomorphism identifies $\pi_3(G)$ with $H_3(G,\mathbb Z)$, the Kronecker pairing
identifies $H^3(G,\mathbb Z)$ with $\mathrm{Hom}(H_3(G,\mathbb Z),\mathbb Z)$,
and the transgression in the universal fibration $G\to EG\to BG$ identifies the
positive generator of $H^4(BG,\mathbb Z)$ with the positive generator of
$H^3(G,\mathbb Z)$. The main point of this section is to make these generators
canonical by fixing compatible normalizations on the Lie algebra side.

\medskip
\noindent
We start by normalizing the basic $\Ad(G)$-invariant inner product on
$\mathfrak g$ using the reduced Killing form, so that long roots have squared
length $2$. This normalization determines a preferred $W$-invariant integral
quadratic form, hence a distinguished generator $y_4(G)\in H^4(BG,\mathbb Z)$
following \cite{Deligne}. On the other hand, the $\mathfrak{sl}_2$-subalgebra associated with
the highest root $\theta$ integrates to an embedded subgroup $SU(2)_\theta\subset G$,
which gives a canonical generator $\gamma_3(G)\in\pi_3(G)$ and therefore a canonical
generator $[S^3_\theta]\in H_3(G,\mathbb Z)$; we then define $x_3(G)\in H^3(G,\mathbb Z)$
by the condition $\langle x_3(G),[S^3_\theta]\rangle=1$. Finally, we explain how the
Cartan $3$-form $\eta_G$ obtained from the same invariant inner product represents
$x_3(G)$, and how the degree--$3$ odd Chern character for any representation produces
a multiple of $x_3(G)$; this yields the ``universal scaling'' statement that all such
constructions differ only by an integer (the Dynkin index) once the basic normalization
is fixed.

\subsection{Integral generators for the fourth cohomology group of BG from the reduced Killing form}

Let $G$ be a compact, simple, simply-connected Lie group with Lie algebra $\mathfrak g$.
Fix a maximal torus $T\subset G$ with Lie algebra $\mathfrak t$, and set
$\mathfrak h=\mathfrak t\otimes_{\mathbb R}\mathbb C\subset \mathfrak g_{\mathbb C}$.

Let $\kappa$ be the Killing form on $\mathfrak g_{\mathbb C}$ and let $h^\vee$ be the dual Coxeter number.
On the compact real form $\mathfrak g$, the form $\kappa$ is negative definite, hence
\begin{equation}\label{eq:basic_inner_product}
(X,Y)
\;:=\;
-\frac{1}{2h^\vee}\,\kappa(X,Y),
\qquad X,Y\in\mathfrak g,
\end{equation}
is positive definite and $\Ad(G)$--invariant\footnote{%
One rigorous way to see the sign is to construct the compact real form from the complex algebra
$\mathfrak g_{\mathbb C}$ using a Chevalley/Weyl basis.
Following Serre \cite[Ch.~VI]{Serre}, there is an antilinear involution $\sigma$ of $\mathfrak g_{\mathbb C}$
(the compact conjugation) characterized on a Cartan subalgebra $\mathfrak h\subset \mathfrak g_{\mathbb C}$ and root
vectors $X_\alpha$ by $\sigma(H)=-H$ for $H\in\mathfrak h$ and $\sigma(X_\alpha)=-X_{-\alpha}$.
Its fixed-point set is the compact real form
\begin{equation}
\mathfrak g
=\bigl(\mathfrak g_{\mathbb C}\bigr)^\sigma
=\operatorname{span}_{\mathbb R}\!\left\{\, iH_\alpha,\; X_\alpha-X_{-\alpha},\; i(X_\alpha+X_{-\alpha}) \,\right\},
\end{equation}
and the Killing form $\kappa$ restricts to a \emph{negative definite} bilinear form on $\mathfrak g$.
Accordingly, when working on the compact form (as in topology and gauge theory) one typically uses the sign-reversed
normalization $-\kappa$ (and in our case $-(2h^\vee)^{-1}\kappa$) to obtain a positive definite $\Ad(G)$-invariant inner
product compatible with the usual Euclidean normalization on the root/coroot lattices.%
}. This is the standard ``basic'' or ``reduced Killing'' normalization.

Let $\Lambda^\vee\subset \mathfrak t$ be the coroot lattice, generated by the coroots $\alpha^\vee$
for roots $\alpha$ of $(\mathfrak g_{\mathbb C},\mathfrak h)$.
Consider the $W$--invariant quadratic form on $\Lambda^\vee$
\begin{equation}
Q(H)\;:=\;\frac12(H,H).
\end{equation}
With the normalization \eqref{eq:basic_inner_product}, one has
\begin{equation}
Q(\alpha^\vee)=1 \qquad\text{for every \emph{long} root $\alpha$}.
\end{equation}

\medskip

\begin{thm}[Deligne {\cite[Lemma 1.7.3]{Deligne}}]\label{thm:Deligne_generator_H4_killing_version}
Let $G$ be compact, simple, and simply-connected. Then $H^4(BG,\mathbb Z)\cong \mathbb Z$, and restriction to $BT$
identifies $H^4(BG,\mathbb Z)$ with the group of $W$--invariant $\mathbb Z$--valued quadratic forms on the coroot lattice.
Under this identification, the \emph{positive} generator $y_4(G)\in H^4(BG,\mathbb Z)$ corresponds to
\begin{equation*}
Q(H)=\frac12(H,H),
\end{equation*}
i.e.\ the unique $W$--invariant integral quadratic form satisfying $Q(\alpha^\vee)=1$ for every long root $\alpha$.
\end{thm}
\medskip

\begin{rem}[Henriques's integrality criterion for $H^4(BG,\mathbb Z)$ and level quantization] (See \cite[Thm.~6]{Henriques}.)
Let $G$ be a compact, connected Lie group with maximal torus $T$ and Lie algebra $\mathfrak t$.
Write
\begin{equation}
\Lambda := \ker\bigl(\exp:\mathfrak t\to T\bigr)
\end{equation}
for the cocharacter lattice.  By Chern--Weil theory, any class in $H^4(BG,\mathbb Z)$ determines an
$\Ad(G)$--invariant inner product $\langle\cdot,\cdot\rangle$ on $\mathfrak g$ (equivalently, a quadratic form
$Q(X)=\frac12\langle X,X\rangle$).

Henriques shows that the image of $H^4(BG,\mathbb Z)$ inside $(S^2\mathfrak g^\ast)^G$ is precisely the lattice of
$\Ad(G)$--invariant inner products satisfying the integrality condition
\begin{equation}
\frac12\langle X,X\rangle \in \mathbb Z
\qquad\text{for all } X\in \Lambda .
\end{equation}
In particular, when $G$ is semisimple and simply-connected one has $\Lambda=\Lambda^\vee$ (the coroot lattice),
and this criterion is equivalent to Deligne's description of $H^4(BG,\mathbb Z)$ as the group of $W$--invariant
integral quadratic forms on $\Lambda^\vee$.  With our normalization $(\theta,\theta)=2$, the resulting primitive
class is exactly the quadratic form
\begin{equation}
Q(H)=\frac12(H,H),
\end{equation}
hence it agrees with our generator $y_4(G)$ and with the corresponding quantization condition for Chern--Simons/WZW
levels.
\end{rem}

\begin{rem}[Harris's $K$--theoretic identification of the basic generator]
The identification of the basic class $y_4(G)\in H^4(BG,\mathbb{Z})$ (and its transgressed image
$x_3(G)\in H^3(G,\mathbb{Z})$) is consistent with, and in fact already \emph{explicit} in, the work of
Harris~\cite{Harris}.
In \cite[\S 3]{Harris}, Harris identifies the generator $y_4(G)$ via the canonical isomorphism
$H^4(BG,\mathbb{Z})\cong H^4(BT,\mathbb{Z})^{W}$ and shows that it is represented by the
$W$--invariant quadratic polynomial attached to the basic $\Ad(G)$--invariant inner product
$(\cdot,\cdot)$ on $\mathfrak{t}$: more precisely, for the primitive quadratic generator $P_1$,
he proves that
\begin{equation*}
\frac{(\beta,\beta)}{4}\,P_1 \;\in\; H^4(BT,\mathbb{Z})^{W}
\end{equation*}
represents a generator of $H^4(BG,\mathbb{Z})$, where $\beta$ is a root of maximal length.
With our normalization $(\theta,\theta)=2$ (equivalently $(\beta,\beta)=2$ for any long root),
this is exactly the class corresponding to the quadratic form
\begin{equation*}
Q(H)=\frac{1}{2}(H,H),
\end{equation*}
hence $Q$ is the basic integral generator detected by the transgression and by odd $K$--theory.
\end{rem}

\subsection{Transgression to third cohomology group of G  and the canonical generator of the third homotopy group}

Consider the universal principal bundle $G\to EG\to BG$. The associated transgression map identifies the
degree--$4$ generator with a degree--$3$ generator on $G$:
\begin{equation}\label{eq:transgression_y4_to_x3}
\tau\big(y_4(G)\big)\;=\;x_3(G)\in H^3(G,\mathbb Z),
\end{equation}
and $H^3(G,\mathbb Z)\cong \mathbb Z$ for $G$ simple and simply-connected.
(Deligne proves this identification in his discussion of the fibration $G\to EG\to BG$.) 

\medskip

\noindent
\textbf{Canonical generator in $\pi_3(G)$.}
For each root $\alpha\in\mathfrak h^\ast$, let $\mathfrak g_\alpha\subset\mathfrak g_{\mathbb C}$ be the root space.
Choose nonzero $X_\alpha\in\mathfrak g_\alpha$ and $X_{-\alpha}\in\mathfrak g_{-\alpha}$ and set
\begin{equation*}
H_\alpha:=[X_\alpha,X_{-\alpha}]\in\mathfrak h.
\end{equation*}
After scaling $X_\alpha,X_{-\alpha}$, we may assume $(X_\alpha,H_\alpha,X_{-\alpha})$ is an $\mathfrak{sl}_2$--triple:
\begin{equation*}
[H_\alpha,X_\alpha]=2X_\alpha,\qquad
[H_\alpha,X_{-\alpha}]=-2X_{-\alpha},\qquad
[X_\alpha,X_{-\alpha}]=H_\alpha,
\end{equation*}
in which case $H_\alpha=\alpha^\vee$.

Intersecting $\langle X_\alpha,H_\alpha,X_{-\alpha}\rangle$ with the compact real form $\mathfrak g$
gives a subalgebra $\mathfrak{su}(2)\subset\mathfrak g$, hence a subgroup $SU_\alpha(2)\subset G$ and an embedding
$$
\phi_\alpha:SU(2)\cong SU_\alpha(2)\hookrightarrow G.
$$
Let $\theta$ be the highest root (hence a long root). We define
\begin{equation*}
\gamma_3(G):=[\phi_\theta]\in\pi_3(G)\cong\mathbb Z
\end{equation*}
to be the \emph{positive} generator. Let
\begin{equation*}
S^3_\theta:=\phi_\theta(SU(2))\subset G,\qquad [S^3_\theta]\in H_3(G,\mathbb Z).
\end{equation*}
Then $[S^3_\theta]$ is a generator of $H_3(G,\mathbb Z)\cong\mathbb Z$, and the positive generator
$x_3(G)\in H^3(G,\mathbb Z)$ is characterized by
\begin{equation}
\langle x_3(G),[S^3_\theta]\rangle = 1.
\end{equation}

\subsection{The Cartan--Koszul 3-form and its relation to the third cohomology group}
\label{subsec:eta_from_reduced_killing}

Let $\theta_L\in\Omega^1(G;\mathfrak g)$ be the left-invariant Maurer--Cartan form.
The invariant inner product \eqref{eq:basic_inner_product} defines (via the Koszul/Cartan construction)
a bi-invariant $3$-form
\begin{equation}\label{eq:etaG_def}
\eta_G
\;:=\;
\frac{1}{12}\,(\theta_L,[\theta_L,\theta_L])
\;\in\;\Omega^3(G),
\end{equation}
where $[\theta_L,\theta_L]=\theta_L\wedge\theta_L$ uses the Lie bracket on $\mathfrak g$.

The Maurer--Cartan equation implies $d\eta_G=0$, hence $[\eta_G]\in H^3(G,\mathbb R)$.
With the basic normalization, the cohomology class is integral; moreover, it represents the
generator $x_3(G)$ once the sign is fixed by requiring
\begin{equation}\label{eq:eta_normalization}
\int_{S^3_\theta}\eta_G = 1.
\end{equation}
Equivalently,
\begin{equation}
[\eta_G]=x_3(G)\in H^3(G,\mathbb Z).
\end{equation}
The $3$-form $\eta_G$ is the Cartan $3$-form attached to the reduced Killing form via the Koszul map.

\subsection{Recovering the Cartan form from the odd Chern character}
\label{subsec:eta_from_odd_chern}

Let $\rho:G\to U(V)$ be a finite-dimensional unitary representation and write
\begin{equation}
\theta_\rho \;:=\; \rho(g)^{-1}d\rho(g)\ \in\ \Omega^1(G;\mathfrak u(V)).
\end{equation}
The \emph{odd Chern character} of the Bott class $\beta(\rho)\in K^{-1}(G)$ has the standard de Rham representative
\begin{equation*}
\mathrm{ch}_{\mathrm{odd}}(\beta(\rho))
=\sum_{k\ge 0}\frac{(-1)^k\,k!}{(2\pi i)^{k+1}}\,
\mathrm{tr}\!\big(\theta_\rho^{\,2k+1}\big)\ \in\ \Omega^{\mathrm{odd}}(G),
\end{equation*}
whose degree--$3$ component is the closed bi-invariant $3$-form
\begin{equation}\label{eq:eta_rho_def}
\eta_\rho \;:=\; \frac{1}{24\pi^2}\,\mathrm{tr}\!\big(\theta_\rho^{\,3}\big)\ \in\ \Omega^3(G),
\qquad d\eta_\rho=0.
\end{equation}
Thus
\begin{equation}
\big(\mathrm{ch}_{\mathrm{odd}}(\beta(\rho))\big)_{(3)} \;=\; [\eta_\rho]\ \in\ H^3(G,\mathbb R).
\end{equation}

\medskip
\noindent\textbf{Scaling by Dynkin index.}
Using the representation $\rho$, define an $\Ad(G)$--invariant symmetric bilinear form on $\mathfrak g$ by
\begin{equation}\label{eq:trace_form_rho}
B_\rho(X,Y)\;:=\;-\mathrm{tr}\!\big(\rho_\ast(X)\rho_\ast(Y)\big),
\qquad X,Y\in\mathfrak g.
\end{equation}
Since $G$ is compact and simple, the space of $\Ad(G)$--invariant symmetric bilinear forms on $\mathfrak g$ is one-dimensional,
so there exists a unique scalar $j_\rho\in\mathbb R$ such that
\begin{equation}\label{eq:dynkin_index_as_scalar}
B_\rho \;=\; j_\rho\,(\cdot,\cdot),
\end{equation}
where $(\cdot,\cdot)$ is the reduced Killing form from \eqref{eq:basic_inner_product}.
With our normalization (long roots have squared length $2$), the scalar $j_\rho$ is a \emph{positive integer} and equals the
Dynkin index of the representation $\rho$.

\medskip

\begin{prop}\label{prop:eta_rho_scales_x3}
Let $G$ be compact, simple, and simply-connected, and let $x_3(G)\in H^3(G,\mathbb Z)\cong\mathbb Z$ be the positive generator
characterized by $\langle x_3(G),[S^3_\theta]\rangle=1$.
Then for every unitary representation $\rho$ one has
\begin{equation}\label{eq:eta_rho_equals_jrho_x3}
[\eta_\rho]\;=\; j_\rho\,x_3(G)\ \in\ H^3(G,\mathbb Z)\subset H^3(G,\mathbb R),
\end{equation}
equivalently,
\begin{equation*}
\int_{S^3_\theta}\eta_\rho \;=\; j_\rho.
\end{equation*}
\end{prop}

\begin{proof}[Proof sketch]
Both $\eta_\rho$ and $\eta_G$ are bi-invariant $3$-forms, hence determined by their values at the identity, i.e.\ by
$\Ad(G)$--invariant alternating $3$-forms on $\mathfrak g$.
For a compact simple $\mathfrak g$, this space is one-dimensional and generated by
$(X,[Y,Z])$ (equivalently, by the Koszul map applied to an invariant bilinear form).
Replacing $(\cdot,\cdot)$ by $B_\rho=j_\rho(\cdot,\cdot)$ scales the resulting Cartan $3$-form by the same factor $j_\rho$.
Evaluating on the highest-root subgroup $S^3_\theta=SU(2)_\theta$ fixes the integral normalization, giving
$\int_{S^3_\theta}\eta_\rho=j_\rho$ and hence \eqref{eq:eta_rho_equals_jrho_x3}.
\end{proof}

\medskip
\noindent\textbf{Recovering $\eta_G$ from a representation.}
If $\rho$ has Dynkin index $j_\rho=1$, then \eqref{eq:eta_rho_equals_jrho_x3} gives $[\eta_\rho]=x_3(G)$ and hence
$\eta_\rho$ is a de Rham representative of the canonical generator.
More generally,
\begin{equation*}
x_3(G)\;=\;\frac{1}{j_\rho}\,[\eta_\rho].
\end{equation*}

\medskip
\noindent\textbf{Example: the adjoint representation.}
For $\rho=\mathrm{Ad}$ one has
\begin{equation*}
B_{\mathrm{Ad}}(X,Y)=-\mathrm{tr}(\mathrm{ad}_X\mathrm{ad}_Y)=-\kappa(X,Y)=2h^\vee\,(X,Y),
\end{equation*}
so $j_{\mathrm{Ad}}=2h^\vee$ and therefore
\begin{equation*}
[\eta_{\mathrm{Ad}}]\;=\;2h^\vee\,x_3(G).
\end{equation*}
This is often a convenient way to compare trace-normalizations appearing in the literature with the reduced Killing normalization.

\section{ Normalizations in a G-bundle}
\label{sec:CS}

This appendix records the standard differential-form realizations of the low-degree classes
\begin{equation*}
y_4(G)\in H^4(BG,\mathbb Z),\qquad x_3(G)\in H^3(G,\mathbb Z),
\end{equation*}
and explains how they arise naturally in Yang--Mills theory, Chern--Simons theory, and Wess--Zumino--Witten terms.

\subsection{Normalized Ad(G)-invariant bilinear form, trace, and quadratic form}
\label{subsec:CS_normalized_trace}

Let $G$ be a compact, connected Lie group with Lie algebra $\mathfrak g$.
Fix an $\Ad(G)$--invariant symmetric bilinear form 
\begin{equation*}
(\cdot,\cdot):\mathfrak g\times\mathfrak g\to\mathbb R.
\end{equation*}
When $G$ is simple, this bilinear form is uniquely defined by requesting that it respects the standard Dynkin's normalization:
\begin{equation}\label{eq:CS_long_root_norm}
(\alpha^\vee,\alpha^\vee)=2
\qquad\text{for every long root $\alpha$.}
\end{equation}
For example, that is the case of the \emph{reduced Killing form}
\begin{equation}\label{eq:CS_basic_form}
(X,Y)\;:=\;-\frac{1}{2h^\vee}\,\kappa(X,Y)
\;=\;-\frac{1}{2h^\vee}\,\mathrm{Tr}\!\big(\ad_X\circ \ad_Y\big),
\qquad X,Y\in\mathfrak g.
\end{equation}
On the compact real form $\mathfrak g$, the Killing form $\kappa$ is negative definite, we therefore introduced an overall negative sign to have a positive-definite bilinear form.

\begin{rem}[Sign convention: algebraic versus topological normalization]
The sign in our definition of the reduced Killing form \eqref{eq:reducedKilling} warrants a comment,
since it does not appear in Dynkin's original algebraic work \cite{Dynkin.SubA,Dynkin59}.
Dynkin works naturally with a complex semisimple Lie algebra $\mathfrak g_{\mathbb C}$ and the Euclidean
space $\mathfrak h_{\mathbb R}^*$ spanned by the roots, equipped with a positive definite $W$-invariant
inner product normalized by $(\theta,\theta)=2$ for the highest (long) root $\theta$.

In contrast, gauge theory and the topology of $G$-bundles are formulated on the compact real form
$\mathfrak g$ (e.g.\ $\mathfrak{su}(n)$ rather than $\mathfrak{sl}(n,\mathbb C)$).
On $\mathfrak g$ the Killing form $\kappa(X,Y)=\Tr(\ad_X\ad_Y)$ is negative definite, so the
standard choice of a \emph{positive definite} $\Ad$-invariant inner product is the sign-reversed
(and rescaled) form
\begin{equation}
(X,Y)_{\mathrm{basic}}
\;:=\;
-\frac{1}{2h^\vee}\,\kappa(X,Y).
\end{equation}
With this convention, the induced form on the root lattice agrees with the usual Euclidean geometry
($(\theta,\theta)=2$ for long roots), and the corresponding Chern--Weil representatives match the
standard integral generators in low-degree topology (equivalently, the basic $SU(2)$ instanton on
$S^4$ has charge $+1$ with the conventional orientation).
\end{rem}

For any nontrivial finite-dimensional complex representation $\rho$ of $G$, let $\ell_\rho\in\mathbb Z_{>0}$ be the Dynkin index of the representation $\rho$, then 
\begin{equation}
\Tr\big(\rho(X)\rho(Y)\big)=-\ell_\rho\,(X,Y).
\end{equation}
We then define the \emph{normalized trace} $\trb$ on $\mathfrak g$ by
\begin{equation}
\trb(XY):=-\frac{1}{\ell_\rho}\Tr\big(\rho(X)\rho(Y)\big).
\end{equation}
This definition is independent of the choice of $\rho$, and satisfies $(X,Y)=\trb(XY)$. We recall that the the Dynkin index for the adjoint representation is $2h^\vee$ with $h^\vee$ the dual Coxeter number of the Lie algebra. 

Let
\begin{equation}\label{eq:CS_Q_def}
Q(X):=\frac12(X,X)=\frac{1}{2}\trb (X^2)
\end{equation}
be the corresponding invariant quadratic polynomial respecting Deligne requirement for integrability of the generator $y_4(G)$ of $H^4(BG,\mathbb{Z})$.

\medskip

\begin{rem}[Brackets versus wedge products for $\mathfrak g$--valued forms]\label{rem:bracket_wedge_convention}
If $A\in\Omega^p(M;\mathfrak g)$ and $B\in\Omega^q(M;\mathfrak g)$, we use the \emph{graded Lie bracket}
\begin{equation}
[A,B]\;:=\;A\wedge B-(-1)^{pq}B\wedge A,
\end{equation}
where the product $A\wedge B\in\Omega^{p+q}(M;\mathfrak g)$ is defined using the Lie bracket in $\mathfrak g$:
if $A=\sum_i a_i\otimes X_i$ and $B=\sum_j b_j\otimes Y_j$ then
\begin{equation}
A\wedge B\;:=\;\sum_{i,j}(a_i\wedge b_j)\otimes [X_i,Y_j].
\end{equation}
An important case is given by $\mathfrak g$--valued $1$-forms:
\begin{equation}
[A,A]=2\,A\wedge A=2\, A^2.
\end{equation}
in particular, the Maurer--Cartan equation is $d\theta_L+\tfrac12[\theta_L,\theta_L]=0$.
Equivalently, in matrix notation (with $A,B$ matrix-valued forms),
$A\wedge B$ uses matrix multiplication and $[A,B]=A\wedge B-(-1)^{pq}B\wedge A$ is the graded commutator.
\end{rem}

\subsection{Chern--Weil degree 4 and the instanton number}
\label{subsec:CS_ChernWeil}

Let $\pi:P\to M$ be a principal $G$--bundle with connection $A$ and curvature
\begin{equation}\label{eq:CS_curvature}
F \;=\; dA+\frac12[A,A]\;=\; dA + A\wedge A.
\end{equation}
Chern--Weil theory assigns to the invariant polynomial $Q$ the closed $4$-form
\begin{equation}\label{eq:CS_cw_4form}
\omega_4(A) \;:=\; \frac{1}{8\pi^2}\,\trb (F\wedge F)\;\in\;\Omega^4(M),
\end{equation}
whose de Rham class depends only on the isomorphism class of $P$.

If $G$ is simple and simply-connected and $(\cdot,\cdot)$ is the reduced Killing form
\eqref{eq:CS_basic_form}--\eqref{eq:CS_long_root_norm}, then Appendix \ref{sec:AppB} identifies the universal class
$y_4(G)\in H^4(BG,\mathbb Z)$ with $Q$.
Consequently, if $c:M\to BG$ is the classifying map of $P$, then
\begin{equation}\label{eq:CS_pullback_y4}
[\omega_4(A)] \;=\; c^*\bigl(y_4(G)\bigr)\ \in\ H^4(M;\mathbb R),
\end{equation}
and in particular $\,[\omega_4(A)]\,$ is integral.

\medskip
\noindent\textbf{Instanton number.}
If $M$ is a closed oriented $4$-manifold, define the characteristic number
\begin{equation}\label{eq:CS_instanton_number}
k(P)\;:=\;\int_M \omega_4(A)
\;=\;\frac{1}{8\pi^2}\int_M \trb (F\wedge F).
\end{equation}
By \eqref{eq:CS_pullback_y4} this equals $\langle c^*y_4(G),[M]\rangle\in\mathbb Z$ and is independent of $A$.
When $M=S^4$, this integer coincides with the class of $P$ under the identification
$\pi_4(BG)\cong \pi_3(G)\cong \mathbb Z$ (cf.\ AHS).

\medskip
\noindent\textbf{Remark (dependence on the chosen ``trace'').}
If one instead uses $(\cdot,\cdot)_\rho$ coming from a representation trace, then
\begin{equation*}
\frac{1}{8\pi^2}\int_M \Tr_\rho (F\wedge F) \;=\; \ell_\rho\,k(P),
\end{equation*}
with $\ell_\rho$ the Dynkin index of the representation $\rho$.  This is exactly the phenomenon emphasized by
Atiyah--Hitchin--Singer: the integer class of the bundle is canonical, while the Chern--Weil integral
reflects the chosen invariant form (e.g.\ adjoint trace, defining trace, etc.).

\begin{rem}[Relation to $c_2$ / $ch_2$]
For any unitary representation $\rho:G\to U(V)$, the invariant polynomial
$X\mapsto -\Tr_\rho(X^2)$ defines a universal class
$c_\rho\in H^4(BG;\mathbb Z)$ (the degree--$4$ Chern--Weil class associated with $\rho$).
If $P\to M$ is a principal $G$--bundle with connection $A$, then
\begin{equation}
\frac{1}{8\pi^2}\int_M \Tr_\rho(F_A\wedge F_A)
\;=\;
\langle c^*(c_\rho),[M]\rangle
\;\in\;\mathbb Z.
\end{equation}
With our basic normalization, one has $c_\rho=\ell_\rho\,y_4(G)$, hence
\begin{equation}
\frac{1}{8\pi^2}\int_M \Tr_\rho(F_A\wedge F_A)
\;=\;
\ell_\rho\,k(P),
\qquad
k(P)=\langle c^*y_4(G),[M]\rangle .
\end{equation}
When $G=SU(n)$ and $\rho$ is the defining representation, $c_\rho$ is the second Chern class
$c_2$ of the associated rank-$n$ complex vector bundle; in general $c_\rho$ is the degree--$4$
characteristic class determined by $\rho$.
\end{rem}

\subsection{The Chern--Simons 3-form and the boundary formula}
\label{subsec:CS_CSform}

On any $3$-manifold (or on the boundary of a $4$-manifold), the $4$-form $\omega_4(A)$ is locally exact.
A standard primitive is the Chern--Simons $3$-form associated with $(\cdot,\cdot)$:
\begin{equation}\label{eq:CS_CS_def}
\CS(A)\;:=\;\frac{1}{8\pi^2}\trb \left( A\wedge dA+\frac13\,(A\wedge [A,A])\right)
=
\frac{1}{8\pi^2}\trb \left(A\wedge dA+\frac23\,A^3\right)
\in\Omega^3(M).
\end{equation}
A direct computation gives the transgression identity
\begin{equation}\label{eq:CS_dCS}
d\,\CS(A)\;=\;\frac{1}{8\pi^2}\,\trb (F\wedge F)\;=\;\omega_4(A).
\end{equation}

If $M$ is an oriented $4$-manifold with boundary $\partial M$, Stokes' theorem yields
\begin{equation}\label{eq:CS_boundary_formula}
\int_M \omega_4(A)\;=\;\int_{\partial M}\CS(A).
\end{equation}

\medskip
\noindent\textbf{Pure gauge on the boundary.}
If on $\partial M$ the connection is gauge-equivalent to the trivial connection,
so that $A|_{\partial M}=g^{-1}dg$ for a smooth map $g:\partial M\to G$, then $F|_{\partial M}=0$ and
$\CS(A|_{\partial M})$ reduces to a universal $3$-form pulled back from $G$.
This is the Wess--Zumino term discussed next.

\subsection{The Cartan (Wess--Zumino--Witten) 3-form on G}
\label{subsec:CS_Cartan3form}

Let $\theta_L\in\Omega^1(G;\mathfrak g)$ be the left-invariant Maurer--Cartan form.
Define the \emph{Cartan invariant $3$-form} (normalized using the same $(\cdot,\cdot)$ as above) by
\begin{equation}\label{eq:CS_eta_def}
\eta_G
\;:=\;\frac{1}{48\pi^2}\, \trb\big(\theta_L\wedge[\theta_L,\theta_L]\big)
=\frac{1}{24\pi^2}\, \trb\big(\theta_L^3\big)
\ \in\ \Omega^3(G).
\end{equation}
Using the Maurer--Cartan equation $d\theta_L+\frac12[\theta_L,\theta_L]=0$, one checks that $d\eta_G=0$.

\medskip
\noindent\textbf{Relation with Chern--Simons for a pure gauge connection.}
If $A=\theta_L$ (the flat Maurer--Cartan connection), then $F=0$ and \eqref{eq:CS_CS_def} gives
\begin{equation}\label{eq:CS_CS_equals_minus_eta}
\CS(\theta_L)\;=\;-\eta_G.
\end{equation}
More generally, if $A=g^{-1}dg$ on a $3$-manifold $W$, then
\begin{equation}
\CS(A)\;=\; -\,g^*\eta_G.
\end{equation}
Consequently, in the ``pure gauge at infinity'' situation for $M$ with boundary,
\eqref{eq:CS_boundary_formula} becomes the familiar winding-number formula
\begin{equation}\label{eq:CS_winding_formula}
k(P)
\;=\;
\int_{\partial M}\CS(A)
\;=\;
-\int_{\partial M} g^*\eta_G.
\end{equation}

\medskip
\noindent\textbf{Integrality and the generator of $H^3(G,\mathbb Z)$.}
If $G$ is compact, simple, and simply-connected and $(\cdot,\cdot)$ is the reduced Killing form,
then $H^3(G,\mathbb Z)\cong \mathbb Z$ and the class of $\eta_G$ is integral.
Fix the \emph{index-one} subgroup $SU(2)_\theta\subset G$ associated with the highest root (Appendix~B),
and normalize the sign by requiring
\begin{equation}\label{eq:CS_eta_normalization}
\int_{SU(2)_\theta}\eta_G = 1.
\end{equation}
Then
\begin{equation}\label{eq:CS_x3_from_eta}
x_3(G)\;:=\;[\eta_G]\ \in\ H^3(G,\mathbb Z)
\end{equation}
is the positive generator, characterized equivalently by
\begin{equation*}
\langle x_3(G),[SU(2)_\theta]\rangle=1.
\end{equation*}

\noindent\textbf{Transgression.}
In the universal bundle $G\to EG\to BG$, the Serre transgression identifies the degree-$4$ generator
with the degree-$3$ generator:
\begin{equation}\label{eq:CS_transgression_relation}
\tau\big(y_4(G)\big)=x_3(G),
\end{equation}
and at the level of differential forms the Cartan form $\eta_G$ is the Chern--Simons transgression form
associated with the invariant polynomial $Q$.

\subsection{Gauge transformations and level quantization}
\label{subsec:CS_gauge}

Let $W$ be a closed oriented $3$-manifold and $A$ a connection on a principal $G$--bundle over $W$.
Under a gauge transformation $g:W\to G$ one has the standard identity
\begin{equation}\label{eq:CS_gauge_transform}
\CS(A^g)=\CS(A)-g^*\eta_G +d\,\Xi,\qquad 
\Xi= \frac{1}{8\pi^2} \trb \left( A \wedge g^*\theta_R \right)
\end{equation}
where $\theta_R$ denotes the Maurer--Cartan right invariant form ($g^*\theta_R=dg g^{-1}$ and satisfying $d\theta_R-\tfrac{1}{2}[\theta_R,\theta_R]=0$). 
In particular, on a closed $W$,
\begin{equation}\label{eq:CS_shift}
\int_W \CS(A^g)\;-\;\int_W \CS(A)\;=\;-\int_W g^*\eta_G \ \in\ \mathbb Z,
\end{equation}
because $[\eta_G]\in H^3(G,\mathbb Z)$ is integral.
Thus the \emph{Chern--Simons functional at integer level $k\in\mathbb Z$}
\begin{equation}\label{eq:CS_level_m_action}
\exp\!\left(2\pi i\,k\int_W \CS(A)\right)\in S^1
\end{equation}
is gauge-invariant.  This is the usual ``level quantization'' statement in Chern--Simons theory,
and it is exactly the integrality of $x_3(G)$ (equivalently $y_4(G)$) expressed in physics language.

\begin{rem}
The allowed integer levels are precisely the lattice of invariant inner products whose
quadratic form $Q(X)=\frac12(X,X)$ is integral on $\Lambda=\ker(\exp:\mathfrak t\to T)$; see
\cite[Thm.~6]{Henriques}.
\end{rem}

\subsection{Link with the Dynkin embedding index}
\label{subsec:CS_scaling_embeddings}

Let $f:G\hookrightarrow H$ be an embedding of compact, simple, simply-connected Lie groups.
By definition of the Dynkin embedding index $j_f$, the pullback of the reduced Killing form of $H$ satisfies
\begin{equation*}
(fX,fY)_H \;=\; j_f\,(X,Y)_G.
\end{equation*}
Consequently, all the associated differential forms scale uniformly:
\begin{align*}
(Bf)^*\bigl(y_4(H)\bigr) &= j_f\,y_4(G),\quad 
f^*\bigl(x_3(H)\bigr) = j_f\,x_3(G),\quad
f^*(\eta_H) = j_f\,\eta_G,\\
\frac{1}{8\pi^2}\trb(F_H\wedge F_H) &= j_f\,\frac{1}{8\pi^2}\trb(F_G\wedge F_G)
\qquad\text{(after identifying the curvatures via $f$).}
\end{align*}
This is the differential-form incarnation of the ``universal scaling'' theorem proved in the main text.

\subsection{Normalization ladder and a summary table}
\label{subsec:CS_summary}

\noindent\textbf{Normalization ladder.}
With our conventions, the basic invariant polynomial $Q(X)=\frac12(X,X)$ produces the integral degree-$4$ class,
whose Chern--Simons transgression produces the integral degree-$3$ class on $G$:
\begin{equation}\label{eq:CS_normalization_ladder}
Q(X)=\frac12(X,X)
\;\longrightarrow\;
\omega_4(A)=\frac{1}{8\pi^2}\trb(F\wedge F)
\;\longrightarrow\;
\CS(A)
\;\longrightarrow\;
\eta_G=\frac{1}{24\pi^2}\trb(\theta_L^3).
\end{equation}
The last arrow means: on a pure gauge $A=g^{-1}dg$ one has $\CS(A)=-g^*\eta_G$
(cf.\ \eqref{eq:CS_CS_equals_minus_eta}).

\begin{table}[htb]
\begin{center}
\renewcommand{\arraystretch}{1.25}
\begin{tabular}{|l|p{7cm}|}
\hline
Invariant bilinear form & $(X,Y)$ on $\mathfrak g$ \\
\hline
Invariant quadratic polynomial & $Q(X)=\frac12(X,X)$ \\
\hline
Chern--Weil $4$-form (degree $4$ class) & $\displaystyle \omega_4(A)=\frac{1}{8\pi^2}\trb(F\wedge F)$ \\
\hline
Chern--Simons $3$-form & $\displaystyle \CS(A)=\frac{1}{8\pi^2}\trb\Big(A\wedge dA+\frac23 A^3\Big)$ \\
\hline
Cartan / WZW $3$-form on $G$ & $\displaystyle \eta_G=\frac{1}{24\pi^2}\trb(\theta_L^3)$ \\
\hline
Koszul--Cartan $3$--cocycle on $\mathfrak g$ & $\displaystyle \mu_3(X,Y,Z)=(X,[Y,Z])$ \\
\hline
Key identities &
$\displaystyle d\,\CS(A)=\omega_4(A)$, \quad  
$\CS(\theta_L)=-\eta_G$,\\
& $[\eta_G]=x_3(G)$,\quad $\tau(y_4(G))=x_3(G)$ \\
\hline
\end{tabular}
\end{center}
\caption{
{\bf Consistent normalization.} All objects listed on the table are consistently defined using the same invariant bilinear form $(\cdot,\cdot)$ normalized such that $(\alpha^\vee,\alpha^\vee)=2$ for any long root $\alpha$. Hence, they are compatible with the generators $y_4(G)$ of $H^4(BG,\mathbb{Z})$ and $x_3(G)$ fixed in Appendix~B).
} \label{Table:Summary}
\end{table}

\noindent\textbf{Lie algebra cocycle versus Cartan $3$-form.}
Given an $\Ad(G)$--invariant bilinear form $(\cdot,\cdot)$ on $\mathfrak g$, the Koszul (Cartan) $3$--cocycle
$\mu_3(X,Y,Z):=(X,[Y,Z])$
is the infinitesimal avatar of the bi-invariant Cartan $3$-form $\eta_G$ on $G$: under the standard identification
$\Omega^3(G)^{\mathrm L}\cong \Lambda^3\mathfrak g^\ast$ by evaluation at the identity, one has
$(\eta_G)_e \;=\; c\,\mu_3$ 
for a constant $c$ fixed by our normalization (e.g.\ $\int_{SU(2)_\theta}\eta_G=1$).

\end{appendices}

\end{document}